\documentclass{cmslatex}
\usepackage{latexsym, amssymb, enumerate, amsmath}

\usepackage{graphicx}
\renewcommand{\theequation}{\arabic{section}.\arabic{equation}}

\makeatletter
    
    \newcommand{\Rmnum}[1]{\expandafter\@slowromancap\romannumeral #1@}
  \makeatother

\def\({\left(}
\def\){\right)}
\def\[{\left[}
\def\]{\right]}

\newtheorem{thm}{Theorem}[section]
\newtheorem{prop}[thm]{Proposition}
\newtheorem{lem}[thm]{Lemma}
\newtheorem{cor}[thm]{Corollary}

\newtheorem{defn}[thm]{Definition}

\newtheorem{rem}[thm]{Remark}

\begin{document}

\title{The Continuum Limit of Toda Lattices for Random Matrices with Odd Weights
}
\author{Nicholas M. Ercolani
\thanks {Department of Mathematics, The University of Arizona, Tucson, AZ 85721--0089, ({\tt ercolani@math.arizona.edu}). Supported by NSF grant DMS-0808059.}
\and Virgil U. Pierce \thanks {Department of Mathematics, The University of Texas -- Pan American, Edinburg, TX 78539, ({\tt piercevu@utpa.edu}).  Supported by NSF grant DMS-0806219.}}



\pagestyle{myheadings} \markboth{CONT. LIMIT OF TODA LATTICES FOR RANDOM MATRICES}{
N. M. ERCOLANI AND V. U. PIERCE}\maketitle

\begin{abstract}
This paper is concerned with the asymptotic behavior of the free energy for a class of Hermitean random matrix models, with odd degree polynomial potential,  in the large $N$ limit. It continues an investigation initiated and developed in a sequence of prior works  whose ultimate aim is to reveal and understand, in a rigorous way,  the deep connections between correlation functions for eigenvalues of these random matrix ensembles on the one hand and the enumerative interpretations of their matrix moments in terms of map combinatorics (a branch of graph theory) on the other.  In doing this we make essential use of the link between the asymptotics of the random matrix partition function and orthogonal polynomials with exponential weight equal to the random matrix potential. Along the way we develop and analyze the continuum limits of both the Toda lattice equations and the difference string equations associated to these orthogonal polynomials. The former are found to have the structure of a hierarchy of  near-conservation laws universal in the potential; the latter are a novel semi-classical extension of the traditional string equations.   Our methods apply to regular maps of both even and odd valence, however we focus on the latter since that is the relevant case for this paper. These methods enable us to rigorously determine closed form expressions for the generating functions that enumerate trivalent maps, in general implicitly, but also explictly in a number of cases.
\end{abstract}

\begin{keywords}
random matrices, Toda lattice, Motzkin paths, string equations, conservation law hierarchies, map enumeration

\smallskip

{\bf subject classifications.} 
05C30, 34M55, 60B20
\end{keywords}

\section{Introduction} \label{sec:1}
The general class of matrix ensembles we want to analyze has probability measures of the form
\begin{eqnarray} \label{RMT}
d\mu_{t_{j}} &=& \frac{1}{{Z}^{(n)}_N(t_{j})}\exp\left\{-N \mbox{ Tr} [V_j(M, t_{j})]\right\} dM,\,\, \mbox{where}\\
\label{I.001b} V_j(\lambda; \ t_{j} ) &=&  \frac{1}{2} \lambda^{2} +  t_{j} \lambda^{j}
\end{eqnarray}
defined on the space $\mathcal{H}_n$ of $n \times n$ Hermitean matrices, $M$, and with $N$ a positive parameter. The normalization factor  ${Z}^{(n)}_N(t_{j})$, which serves to make $\mu_t$  a probability measure, is called the {\it partition function} of this unitary ensemble. Previous works, \cite{EM03, EMP08, KP09, Er09}, have focussed on the case of {\it even} $j$ for which the measure (\ref{RMT}) is indeed normalizable for $t_j > 0$. The case of {\it odd} $j$ is more complicated; it is not clear prima facie how to initiate a rigorous analysis in this setting. 

Very recently, however, a generalization of the {\it equilibrium measure} (which governs the leading order behavior of the free energy associated to (\ref{RMT}))  was developed and applied to this problem, \cite{BD10}. It is based on a complex contour  deformation of the variational problem for the leading order of the free energy that was motivated by new ideas in approximation theory related to complex Gaussian quadrature of integrals with high order stationary points \cite{DHK}. 

This analysis shows that an equilibrium measure associated to the weight $\exp\left\{-N\left(\frac{1}{2}\lambda^2 + t_{2\nu + 1} \lambda^{2\nu+1} \right)\right\}$, with dominant exponent odd, will exist. It is constructed explicitly for the case of a cubic weight, $\nu = 1$, in \cite{BD10}.  A detailed study of the explicit construction for general odd-dominant weights, as opposed to just the existence argument which may be deduced from \cite{DHK}, will be taken up elsewhere.
The boundaries of the support of this equilibrium measure are determined by the simultaneous solutions of the two equations:
\begin{align}
\int_{A}^B \frac{V'(\lambda)}{\sqrt{(\lambda-A)(\lambda-B)}} d\lambda &= 0 \\
\int_{A}^B \frac{\lambda V'(\lambda)}{\sqrt{(\lambda-A)(\lambda-B)}} d\lambda &= 2 \pi i .
\end{align}

One can compute these integrals which, in the cubic case, $V(\lambda) = \frac{1}{2} \lambda^2 + t_3 \lambda^3 $, leads to a pair of equivalent algebraic equations determining $A$ and $B$.
\begin{align}
\frac{1}{2} (A + B) + 3 t_3 ( \frac{3}{8} A^2 + \frac{1}{4} AB + \frac{3}{8} B^2 ) &= 0 \\
( \frac{3}{8} A^2 + \frac{1}{4} AB + \frac{3}{8} B^2 ) + 3 t_3 \left( \frac{5}{16} A^3 + \frac{3}{16} A^2 B + \frac{3}{16} A B^2 + \frac{5}{16} B^3 \right) &= 2 \,.
\end{align}

It is natural to make the following change of variables:  $z_0 = \frac{1}{16} \left(B-A\right)^2$ and $u_0 =\frac{1}{2}\left( A + B\right)$.  
The corresponding algebraic equations for $z_0$, $u_0$ are
\begin{align} \label{alg.1}
u_0 + 3 t_3 ( u_0^2 + 2 z_0) &= 0 \\ \label{alg.2}
u_0^2 + 2z_0 + 3 t_3 ( u_0^3 + 6 u_0 z_0) &= 2 \,.
\end{align}

With the above notations, the interval of support of the equilibrium measure in the cubic case may be written as
\begin{eqnarray*}
[A,B] &=& \left[ u_0 -2\sqrt{z_0}, u_0 + 2\sqrt{z_0}\, \right].
\end{eqnarray*}
The equilibrium measure has a variational characterization \cite{DHK} and from the variational equations the measure can be explicitly determined to be
\begin{eqnarray}\label{eqmeas}
\frac{1}{2\pi i} \left(1 + 3t_3(\lambda +  u_0)\right) \chi_{[A,B]}(\lambda)\sqrt{(\lambda - A)(\lambda - B)}.
\end{eqnarray}
A minimal basis for the ideal of relations given by (\ref{alg.1}) and (\ref{alg.2}) is 
\begin{eqnarray}
\label{ideal1} 3t_3 u_0^2 + u_0 +6t_3z_0 &=& 0\\
\label{ideal2} - 6t_3z_0 u_0+ (1- z_0) &=& 0\,.
\end{eqnarray}
It is straightforward to use (\ref{ideal1}) to eliminate $z_0$ in (\ref{ideal2}) and get
\begin{eqnarray}
\label{f_eqn} 18 t_3^2 u_0^3 + 9t_3u_0^2 + u_0  + 6 t_3 = 0\,.
\end{eqnarray}
The resultant of  (\ref{alg.1}) and (\ref{alg.2}) eliminating $u_0$ is given by 
\begin{eqnarray}
R(z_0) &=& 
\left| 
\begin{array}{ccc}
3t_3  & 1  & 6t_3 z_0  \\
-6 t_3 z_0  &  1-z_0  & 0   \\
0  &  -6t_3 z_0   & 1-z_0   
\end{array}
\right|\\
\label{g_eqn}&=& 3 t_3 \left(72 t_3^2 z_0^3 - z_0^2 + 1 \right) = 0\,.
\end{eqnarray}
We note that (\ref{g_eqn}) has a form that is reminiscent of the implicit equation for $z_0$ that one has in the case of even weights \cite{EMP08, Er09}.

For general polynomial potentials $V$ with even dominant power, it is possible to establish the following fundamental asymptotic expansion \cite{EM03}, \cite{EMP08} of the free energy associated to the partition function.  More precisely, those papers consider potentials of the form 
\begin{eqnarray}\label{genpot}
V(\lambda) &=& \frac{1}{2} \lambda^{2} + \sum_{j=1}^{J} t_{j} \lambda^j,
\end{eqnarray}
with $J=2\nu$.  Introducing a renormalized partition function, which we refer to as a {\it tau function},
\begin{equation} \label{tausquare}
\tau^2_{n,N}(\vec{t}\;) =  \frac{Z^{(n)}_{N}(\vec{t}\;)}{Z^{(n)}_{N}(0)},
\end{equation}
where $\vec{t} = (t_1, \dots t_{J}) \in \mathbb{R}^J$, this expansion has the form

\begin{eqnarray}
\label{I.002} \ \ \ \log \tau^2_{n,N}(\vec{t}\;) =
n^{2} e_{0}(x, \vec{t}\;) + e_{1}(x, \vec{t}\;) + \frac{1}{n^{2}} e_{2}(x, \vec{t}\;) + \cdots +  \frac{1}{n^{2g-2}} e_{g}(x, \vec{t}\;) + \dots
\end{eqnarray}
as $n, N \to \infty$ with $x = \frac{n}{N}$ held fixed. Moreover, for $\mathcal{T} =  (1-\epsilon, 1 + \epsilon) 
\times \left(\{|\vec{t}\;| < \delta\} \cap \{t_{J}> 0\}\right) $ for some $\epsilon > 0, \delta > 0$,
\begin{enumerate}[(i)]

\item \label{unif} the expansion is uniformly valid on a compact subsets of $\mathcal{T}$;

\item \label{analyt} $e_g(x, \vec{t}\;)$ extends to be complex analytic in 
$\mathcal{T}^{\mathbb{C}} =\left\{(x,\vec{t}\;) \in \mathbb{C}^{J+1} \big| |x - 1| < \epsilon, |\vec{t}| < \delta\right\}$;

\item \label{diff} the expansion may be differentiated term by term in $(x,\vec{t}\;)$ with uniform error estimates as in (\ref{unif})

\end{enumerate}
The meaning of (\ref{unif}) is that for each $g$ there is a constant, $K_g$, depending only on $\mathcal{T}$  and $g$ such that 
\begin{equation*}
\left| \log \tau^2_{n,N} \left(\vec{t}\;\right) - n^2 e_0(x, \vec{t}\;)  -  \dots  - \frac{1}{n^{2g-2}} e_{g}(x, \vec{t}\;) \right|  \leq \frac{K_g}{n^{2g}}
\end{equation*}
for $(x, \vec{t}\;)$ in a compact subset of $\mathcal{T}$. The estimates referred to in (\ref{diff}) have a similar form with $\tau^2_{n,N}$ and $e_{j}(x, \vec{t}\;)$ replaced by their mixed derivatives (the same derivatives in each term) and with a possibly different constant.
\smallskip

This result is based on the analysis of a Riemann-Hilbert problem (RHP) for orthogonal polynomials on $\mathbb{R}$ whose exponential weight is associated to the weight of the random matrix measure. This  RHP was first introduced in \cite{FIKII} for studying the asymptotic behavior of random matrix partition functions. The relevant analysis of this RHP for the above result was carried out in \cite{EM03} by the method of {\it nonlinear steepest descent} \cite{Deift}. In particular, in \cite{EM03} it is shown that the constants $K_g$ are explicitly determinable in terms of Airy asymptotics stemming from the Airy parametrix that is used in the vicinity of the endpoints of the support of the equilibrium measure to explicitly solve the RHP.
 
This result extends directly to the case of $V$ with odd dominant power  (i.e. with $J = 2\nu+1$) once one has the existence of the equilibrium measure (which is  explicitly given by (\ref{eqmeas}) in the cubic case). This involves studying the asymptotic behavior of the appropriate non-Hermitean orthogonal polynomials for the given odd weight (see section \ref{sec:2}). More precisely, the Riemann-Hilbert analysis of \cite{EM03} carries over mutatits mutandis to a Riemann-Hilbert problem for the non-Hermitean orthogonal polynomials with the principal difference being that the contour along which the jump matrices are originally defined is no longer the real axis but rather a deformed contour \cite{DHK, BD10}. 
\bigskip

Our principal interest in this paper is to better understand the analytical structure of the coefficients $e_g$ for potentials of the form (\ref{genpot}) when $J$ there is odd. For $J$ even these coefficients provide a wealth of information about problems in combinatorial enumeration as well as about eigenvalue correlations for random matrices \cite{Er09}. One expects to see similar connections in the case of odd $J$ but this is much less developed.

 Despite the fact that we are focussed on potentials of the form (\ref{I.001b}) that only depend on a single $t_j$, we will need to appeal to properties (\ref{unif} - \ref{diff}) for other parameters as well, specifically $t_1$ and $x$. That is because in characterizing the $e_g$ we will want to make use of differential relations in 
$t_1$ and $t_j$ between these coefficients as well as certain rescalings of these variables in terms of $x$: 
\begin{eqnarray}
\label{t1-scaing} s_1 &=& x^{-\frac12}t_1\\
\label{tj-scaling} s_j  &=& x^{\frac j2 -1}t_j.
\end{eqnarray}

We will also be studying the tau-functions (\ref{tausquare}) as functions of lattice variables on the non-negative integers, indexed by $n$, which depend analytically on $(x, \vec{t}\;)$ as parameters.  Certain logarithmic derivatives of the tau functions with respect to these parameters satisfy difference equations which, in this context, we refer to as {\it difference string equations}. Furthermore they satisfy differential (in  $t_j$) - difference (in $n$) equations classically known as the {\it Toda lattice equations}. Of particular relevance for describing and analyzing the $e_g$ will be the continuum limits of the difference parts of all these equations. These involve, as independent variables, $s_j$ and $s_1$ as well as a continuous "spatial" variable $w$ in terms of which the differencing in our string and Toda equations may be regarded as a discretization. At the final stage, after one has recursively solved the continuum limit hierarchies for $e_g$ (and various of their derivatives) as functions of $(x, s_1, s_j, w)$, the auxiliary variables are set to $(x, s_1, w) = (1, 0, 1)$ to arrive at the desired closed formulae for $e_g(s_j) = e_g(t_j)$. (Note that when $x=1, s_j = t_j$.) 

The continuum limit of the difference string equations is a hierarchy of nonlinear odes in $w$ while that of the Toda equations is a hierarchy of quasi-linear pdes in $s_j$ and $w$ where in both cases $g$ indexes the respective hierarchy.  One will have these hierarchies for each value of  $j$. 
\bigskip

We take a moment here to briefly explain the connection of the expansion (\ref{I.002}) to combinatorial enumeration that was alluded to earlier.   The $e_g(t_j) = e_g\left(x = 1 , \vec{t} = (0, \dots, 0, t_j) \right)$ (we have set $J = j$ here) are generating functions for the enumeration of {\it $j$-regular maps}.  A map is an embedding of a labelled graph into a compact, oriented and connected  surface $X$ with the requirement that the complement of the graph in $X$ should be a disjoint union of simply connected open sets. More specifically the asymptotic expansion coefficient $e_g$ is a generating function for enumerating  (topological) equivalence classes of maps on a Riemann surface of genus $g$ ($g$-maps) whose embedded graphs are $j$-regular: 
\begin{equation} \label{genusexpA}
e_g(t_j) =  \sum_{m \geq 1}\frac{1}{m!} (-t_j)^{m}\kappa^{(j)}_g(m)
\end{equation}
in which each of the Taylor expansion coefficients $\kappa^{(j)}_g(m)$ is the number of $g$-maps with $m$  $j$-valent vertices. Consequently the Taylor coefficients of $e_g(t_j),  \kappa^{(j)}_g(m)$,  are non-negative integers.  

The notion of $g$-maps was introduced by Tutte and his collaborators in the '60s \cite{Tu} as a means to study the four color conjecture. However,
this subject soon took on a life of its own as a sub-topic of combinatorial graph theory. In the early '80s a group of physicists \cite{BIZ80} discovered a profound connection between the enumerative problem for labelled $g$-maps and diagrammatic expansions of random matrix theory. That seminal work was the basis for bringing asymptotic analytical methods into the study of maps and other related combinatorial problems. 

The {\it trivalent} case of map enumeration (which corresponds to the random matrix ensemble with cubic weight) is of particular relevance for problems in discrete geometry since the corresponding maps are dual to triangulations which are stable discretizations of the associated Riemann surfaces. In particular, in Section \ref{sec:55}, in order to explicitly solve the continuum difference string equations we will use the fact that the coefficient $z_0$ appearing in the equilibrium measure is the generating function for enumerating labelled trivalent maps on a sphere connected to two marked univalent vertices. This is dual to the generating function that enumerates ordered triangulations of the sphere with two marked loops. One also expects $z_0$ to have a combinatorial interpretation analogous to the one it has in the cases of even valence which is as a generating function for the Catalan numbers (in the case of valence 4) and generalized Catalan numbers in the cases of higher even valence \cite{EMP08, Er09}. That interesting topic will be taken up elsewhere.

Recently, in \cite{Er09}, closed form expressions for all of the $e_g(t_j)$ were derived for the cases of even $j$. This was based on the continuum limit of Toda lattice equations, developed in \cite{EMP08}, that are closely related to the random matrix ensembles (\ref{RMT}). In this paper we will derive the analogous continuum $j$-Toda equations for odd $j$ as well as the related hierarchy of continuum difference string equations. 
From these we will derive explicit closed form expressions for the $e_g(t_3)$ for some low values of the genus $g$.
For example we will find that 
\begin{align} 
\nonumber e_0(t_3) &= \frac{1}{2} \log(z_0) + \frac{1}{12} \frac{ (z_0-1)(z_0^2-6 z_0 - 3)}{(z_0+1)}\,, \\
\label{egs} e_1(t_3) &= -\frac{1}{24} \log\left( \frac{3}{2} - \frac{z_0^2}{2} \right)\,, \\
\nonumber e_2(t_3) &= \frac{1}{960} \frac{ (z_0^2-1)^3 (4 z_0^4 - 93 z_0^2 -261)}{(z_0^2 - 3)^5}\,,
\end{align}
where $z_0$ is given by the boundary of the equilibrium measure discussed above and is implicitly related to $t_3$ by the polynomial equation (\ref{g_eqn})
\begin{equation*}
1 = z_0^2 - 72 t_3^2 z_0^3 \,.
\end{equation*}
We expect that our methods will ultimately enable one to derive closed form expressions for the $e_g(t_j)$ for all odd $j$ and all genus $g$.  In \cite{BD10} the Taylor coefficients of $e_0$ and $e_1$ are calculated for the cubic case ($j=3$). These derivations were based on a different approach using the classical string equations. 
\bigskip

The outline of this paper is as follows. In section \ref{sec:2} we summarize the necessary background on non-Hermitean orthogonal polynomials that is the basis for the validity of the asymptotic expansions that we study as well as their continuum limits. This section also presents a {\it path formulation} for both the Toda lattice equations and the difference string equations associated to these orthogonal polynomials. The latter in particular represent a novel method for the study of random matrix continuum limits. The continuum limits themselves are derived in Section \ref{sec:44},  for general odd valence, at least for the leading order and higher order homogeneous terms. To help make this paper more self-contained, a preliminary subsection of Section \ref{sec:44} is included that summarizes prior results on which the work in this paper is based. This part also contains a new result: a description of the asymptotic structure of the diagonal recursion coefficients for the orthogonal polynomials. This result was not needed previously because these diagonal coefficients vanish in the case of even potentials. Although this result has a similar character to what had previously been found for the off-diagonal recursion coefficients, the derivation is technically more complicated due to the fact that the Hirota expression (\ref{anN-preexpansion}) for the diagonal coefficients is given in terms of a leading order differential-difference operator rather than the pure second derivative (\ref{bnN-preexpansion}) for the off-diagonal coefficients. In Section \ref{sec:55} we specialize to the trivalent case and derive the full Toda and difference string equations up to order $g=1$. Moreover, we illustrate the use of these methods by explicitly solving the $g=1$ difference string equations.  Finally in Section \ref{sec:4} we derive a recursive method for expressing each generating function $e_g(t_3)$ in terms of $z_0$ and use this  to establish the explicit formulae (\ref{egs}).

\section{The Role of Orthogonal Polynomials and their Asymptotics} \label{sec:2}

Let us recall the classical relation between orthogonal polynomials and the space of square-integrable functions on the real line, $\mathbb{R}$, with respect to exponentially weighted measures. In particular, we want to focus attention on weights that correspond to the random matrix potentials, $V(\lambda)$, (\ref{I.001b}), that interest us here.  To that end we consider the Hilbert   space $H = L^2\left(\mathbb{R}, e^{-NV(\lambda)}\right)$ of weighted square integrable functions. This space has a natural polynomial basis,
$\{\pi_n(\lambda)\}$, determined by the conditions that
\begin{eqnarray*}
\pi_n(\lambda) &=& \lambda^n + \,\,\mbox{lower order terms}\\
\int \pi_n(\lambda) \pi_m(\lambda) e^{-NV(\lambda)} d\lambda &=& 0\,\, \mbox{for}\,\, n \ne m.
\end{eqnarray*}
For the construction of this basis and related details we refer the reader to \cite{Deift}.

With respect to this basis, the operator of multiplication by $\lambda$ is representable as a semi-infinite tri-diagonal matrix, 
\begin{equation} \label{multop}
\mathcal{L} = \begin{pmatrix} a_0 & 1 &  \\
                              b^2_1 & a_1 & 1 \\
			          & b^2_2 & a_2  & \ddots \\
                                  &      & \ddots & \ddots 
\end{pmatrix}\,.
\end{equation}
$\mathcal{L}$ is commonly referred as the {\it recursion operator} for the orthogonal polynomials and its entries as {\it recursion coefficients}. We remark that often a basis of orthonormal, rather than monic orthogonal, polynomials is used to make this representation.  In that case the 
analogue of (\ref{multop}) is a symmetric tri-diagonal matrix. As long as the coefficients $\{b_n\}$ do not vanish, these two matrix representations can be related through conjugation by a semi-infinite diagonal matrix of the form $\mbox{diag}\,(1, b^{-1}_1, \left(b_1 b_2\right)^{-1}, \left(b_1 b_2 b_3 \right)^{-1}, \dots)$. 

Similarly, the operator of differentiation with respect to $\lambda$, which is densely defined on ${H}$, has a semi-infinite matrix  representation, 
$\mathcal{D}$, which we now determine. Observe that
\begin{eqnarray}
\nonumber \int \pi_n^\prime(\lambda) \pi_m(\lambda) e^{-NV(\lambda)} d\lambda &=& 0 \,\, \mbox{for}\,\, n \leq m;\\
\nonumber \int \pi_n^\prime(\lambda) \pi_m(\lambda) e^{-NV(\lambda)} d\lambda &=& N \int \pi_n(\lambda) V^\prime(\lambda) \pi_m(\lambda) e^{-NV(\lambda)} d\lambda\,\, \mbox{for}\,\, n > m\\
\label{diffrep1} &=&  N \int \pi_n(\lambda) \left\{\lambda + jt \lambda^{j-1} \right\} \pi_m(\lambda) e^{-NV(\lambda)} d\lambda;
\end{eqnarray}
hence,
\begin{eqnarray}
\label{diffrep2}\mathcal{D} &=&  N  \left(\mathcal{L} + j t \mathcal{L}^{j-1}\right)_-
\end{eqnarray}
where the "minus" subscript denotes projection onto the strictly lower part of the matrix. 

From the canonical (Heisenberg) relation on $H$, one sees that
\begin{eqnarray*}
\left[\partial_\lambda, \lambda\right] &=& 1,
\end{eqnarray*}
where here $\lambda$ in the bracket  and $1$ on the right hand side are regarded as multiplication operators. Using this and orthogonality one has
\begin{eqnarray*}
\int \left\{\left[\partial_\lambda, \lambda\right]  \pi_n(\lambda)\right\} \pi_m(\lambda) e^{-NV(\lambda)} d\lambda &=& \kappa_n \delta_{nm},
\end{eqnarray*}
with $\kappa_n > 0$;
\begin{eqnarray*}
&=& \int  \pi_n(\lambda) \left\{\left[\lambda, -\partial_\lambda \right] \pi_m(\lambda) e^{-NV(\lambda)} \right\} d\lambda,
\end{eqnarray*}
where we note that under this {\it transposition} of the bracket within the inner product the order of composition of the operators has interchanged and the minus sign on the derivative comes from integrating by parts;
\begin{eqnarray*}
&=& \int  \pi_n(\lambda) \sum_\ell \left[  \mathcal{L}, \mathcal{D}\right]_{\ell,m} \pi_\ell(\lambda) e^{-NV(\lambda)} d\lambda,\\
\end{eqnarray*}
by (\ref{diffrep1}) and (\ref{diffrep2});
\begin{eqnarray*}
&=& \kappa_n \left[ \mathcal{L}, \mathcal{D} \right]_{n,m}.
\end{eqnarray*}
It follows that 
\begin{eqnarray*}
\left[  \mathcal{L}, \mathcal{D}\right] = I.
\end{eqnarray*}

From this observation one deduces a {\it fundamental relation} among the recurrence coefficients,
\begin{eqnarray} \label{fl1} 
\left[ \mathcal{L}, \left(\mathcal{L} + j t \mathcal{L}^{j-1}\right)_-  \right] &=& \frac{1}{N} I .
\end{eqnarray}
The relations implicit in (\ref{fl1}) have been referred to as {\it string equations} in the physics literature, but their origins go further back to the classical literature in approximation theory \cite{Mag}. In fact the relations that one has, row by row, in (\ref{fl1}) are actually successive differences of consecutive string equations in the usual sense. However, by continuing back to the first row one may recursively de-couple these differences to get the usual equations. To make this distinction clear we will refer to the row by row equations that one has directly from (\ref{fl1}) as {\it difference string equations}.

$\mathcal{L}$ depends smoothly on the coupling parameter $t_j$ in the potential $V(\lambda)$ (see \ref{I.001b}). The explicit dependence can be determined from the fact that multiplication by $\lambda$ commutes with differentiation by $t_j$ and the following consequence of the orthogonality relations:
\begin{align*}
  \int  \frac{\partial}{\partial t_j}\left( \pi_n(\lambda)\right) \pi_m(\lambda) e^{-NV(\lambda)} d\lambda 
&= N \int  \lambda^{j} \pi_n(\lambda)  \pi_m(\lambda) e^{-NV(\lambda)} d\lambda\,,\,\, \mbox{for}\,\, n > m.
\end{align*}
This yields our {\it second fundamental relation} on the recurrence coefficients,
\begin{eqnarray} \label{fl2}
\frac{1}{N}  (\mathcal{L})_{t_{j}} &=& \left[\left(\mathcal{L}^{j}\right)_- , \mathcal{L}\right]\,,
\end{eqnarray}
which is equivalent to  the $j^{th}$ equation of the semi-infinite Toda Lattice hierarchy.

\subsubsection{Odd Weights and Non-Hermitean Orthogonal Polynomials} When $j$ is odd, $H$ as defined above ceases to be a finite measure space; however, by deforming the real axis to an appropriate complex contour  $\Gamma$  one can define a non-Hermitean analogue of orthogonal polynomials with respect to this contour and weight, \cite{BD10, DHK}. These polynomials may not be defined for all values of $n$ but asymptotically they exist (i.e., for $n \geq n_0 - 1$ for some sufficiently large integer $n_0$) \cite{DHK}. Thus one can work on the space   $H = L^2\left(\Gamma, e^{-NV(\lambda)}\right)$ of weighted square integrable functions on the deformed contour  $\Gamma$.  Of course in doing this deformation one can no longer relate the construction of the orthogonal polynomials to an inner product on $H$ as was done before (hence the nomenclature {\it non-Hermitean}). Instead one works with the non-degenerate complex-valued bilinear form that integration naturally gives us. One can then, as we will shortly see, define a basis of polynomials whose pairwise product integrates to zero if they are of different degree. One can still use this basis to represent the recurrence operators and related operators through the bilinear form. To that end let
\begin{equation*}
H = \left\{\mbox{span of}\,\, 1, \lambda, \dots \lambda^{n_0 - 1 }, \pi_{n_0}(\lambda), \pi_{n_0 + 1}(\lambda), \dots \right\}
\end{equation*} 
where $\pi_n(\lambda)$ is a monic polynomial of degree $n$ such that
\begin{equation*}
0 = \int_\Gamma \pi_n(\lambda) \lambda^k e^{-NV(\lambda)} d \lambda \,\, \mbox{for}\,\, k = 0, \dots, n-1.
\end{equation*}
With respect to this basis, multiplication by $\lambda$ is represented as
\begin{equation}
\mathcal{L} = \left( \begin{array}{c|c}
 \mbox{\resizebox{0.5cm}{!}{$\star$}} &
\mbox{\resizebox{0.25cm}{!}{0}} \\ \hline
\begin{matrix}
\begin{matrix} \alpha_0 & \alpha_1 & \dots & \alpha_{n_0-1} \end{matrix} \\
\vspace{0.18cm} \\ \mbox{\resizebox{0.25cm}{!}{0}} \\ \vspace{0.18cm}
\end{matrix}
&
\begin{matrix}
 a_{n_0 } & 1 &  \\
 b^2_{n_0 + 1} & a_{n_0+1} & 1 \\
           & b^2_{n_0 + 2} & a_{n_0+2}  & \ddots \\
           &      & \ddots & \ddots
\end{matrix}
\end{array}\right),
\end{equation}
 where
$b^2_{n_0} \pi_{n_0-1}(\lambda) = \alpha_{n_0-1} \lambda^{n_0-1} + \dots + \alpha_{1} \lambda + \alpha_{0}.$

One may apply standard methods of orthogonal polynomial theory \cite{Szego} to the lower right semi-infinite block of this matrix
\begin{equation} \label{nHmultop}
\hat{\mathcal{L}} = \begin{pmatrix} a_{n_0 } & 1 &  \\
                              b^2_{n_0 + 1} & a_{n_0 + 1} & 1 \\
			          & b^2_{n_0 + 2} & a_{n_0 + 2}  & \ddots \\
                                  &      & \ddots & \ddots 
\end{pmatrix}\,.
\end{equation}
In particular there is a unique semi-infinite lower unipotent matrix $A$ such that
 \begin{eqnarray*}
\hat{\mathcal{L}} &=& A^{-1}\epsilon A
\end{eqnarray*}
where
\begin{eqnarray*}
\epsilon &=& \begin{pmatrix} 0 & 1 &  \\
                              0 & 0 & 1 \\
			          & 0 & 0  & \ddots \\
                                  &      & \ddots & \ddots 
\end{pmatrix}\,.
\end{eqnarray*}
(For a description of the construction of such a unipotent matrix we refer the reader to Proposition 1 of \cite{EM01}.) 

This is related to the Hankel matrix 
\begin{eqnarray*}
\mathcal{H} &=&  \begin{pmatrix} c_{0} & c_1 &  c_2 & \dots\\
                              c_1 & c_2 & c_3 & \dots \\
			      c_2    & c_3 & c_4  & \dots \\
                                \vdots &  \vdots   & \vdots & \ddots 
\end{pmatrix},
\end{eqnarray*}
where
\begin{eqnarray*}
c_k &=& \int_\Gamma \lambda^k e^{-NV(\lambda)} d \lambda
\end{eqnarray*}
is the $k^{th}$ moment of the measure, by
\begin{eqnarray*}
A D A^{\dagger} &=& \begin{pmatrix} c_{n_0 + 1} & c_{n_0 + 2} &  c_{n_0 + 3} & \dots\\
                              c_{n_0 + 2} & c_{n_0 + 3} & c_{n_0 + 4} & \dots \\
			      c_{n_0 + 3}    & c_{n_0 + 4} & c_{n_0 + 5}  & \dots \\
                                \vdots &  \vdots   & \vdots & \ddots 
\end{pmatrix}, \\
D &=& \mbox{diag}\,\left\{ d_{n_0 + 1}, d_{n_0 + 2} \dots \right\}
\end{eqnarray*}
with
\begin{eqnarray*}
d_n &=& \frac{\det \mathcal{H}_n}{\det \mathcal{H}_{n-1}}
\end{eqnarray*}
where $\mathcal{H}_n$ denotes the $n \times n$ principal sub-matrix of $\mathcal{H}$ whose determinant may be expressed as (see Szeg\"o's classical text \cite{Szego}),
\begin{eqnarray}
\nonumber \det \mathcal{H}_n &=& n! \hat{Z}_N^{(n)} \left(t_1,  t_{2\nu+1} \right)\\
\label{szego} \hat{Z}_N^{(n)} \left(t_{2\nu+1}\right) &=& \int_\Gamma \cdots \int_\Gamma \exp\left\{
-N^{2}\left[\frac{1}{N} \sum_{m=1}^{n} V(\lambda_{m}; t_1, \ t_{2\nu + 1})  - \right.\right.
\nonumber 
\\
&& \phantom{\int_\Gamma \cdots \int_\Gamma \exp}\hspace{2cm}
\left. \left.
\frac{1}{N^{2}} \sum_{m\neq \ell} \log{|
\lambda_{m} -
\lambda_{\ell} | } \right] \right\}  d^{n} \lambda,
\end{eqnarray}
where $V(\lambda; t_1, \ t_{2\nu + 1}) = \frac12 \lambda^2 + t_1 \lambda +  t_{2\nu+1} \lambda^{2\nu+1}$.
\medskip

\begin{rem} As mentioned in the introduction, we will sometimes need to extend the domain of the tau 
functions to include other parameters, such as $t_1$, as we have done here. Doing this presents no difficulties in the prior constructions. 
\end{rem}\smallskip

The diagonal elements may in fact be expressed as
\begin{eqnarray*}
d_n &=& \frac{\tau^2_{n, N}}{\tau^2_{n-1, N}} d_n(0)
\end{eqnarray*}
where 
\begin{eqnarray}
\label{szego1}\tau^2_{n, N} &=& \frac{\hat{Z}_N^{(n)}\left(t_1, t_{2\nu+1}\right)}{\hat{Z}_N^{(n)}\left(0,0\right)}\\ 
\label{szego2} &=& \frac{{Z}_N^{(n)}\left(t_1, t_{2\nu+1}\right)}{{Z}_N^{(n)} \left(0,0\right)}
\end{eqnarray}
which agrees with the definition of the tau function given in (\ref{tausquare}). The second equality follows by reducing the unitarily invariant  matrix integrals in (\ref{szego2}) to their diagonalizations which yields (\ref{szego1}) \cite{EM03}. This also provides the connection to eigenvalue correlations alluded to in the introduction.
Tracing through these connections, from $\hat{\mathcal{L}}$ to $D$, one may derive the basic identity relating the random matrix partition function to the recurrence coefficients,
\begin{eqnarray}\label{Hirota}
b^2_{n,N} &=& \frac{\tau^2_{n+1, N}\tau^2_{n-1, N}}{\tau^4_{n, N}} b^2_{n,N}(0)
\end{eqnarray}
which is the basis for our analysis of continuum limits in the next section. With this, the fundamental relations (\ref{fl1}) and (\ref{fl2}) continue to hold in the non-Hermitean case for $n$ sufficiently large.
\bigskip

\begin{rem} It needs to be noted that the bilinear form used to define orthogonal polynomials and recurrence coefficients in this section depends on the choice of the contour $\Gamma$ and therefore so do these polynomials and coefficients. However, it does not affect the asymptotics of these objects. This is a consequence of the fact that outside the locus of support of the equilibrium measure, one has exponential decay of the asymptotics. The deformation of $\Gamma$ away from $\mathbb{R}$ is taken in these exponentially decaying regimes. We refer the reader to \cite{FIKII, DK, Er09} where similar issues concerning non-Hermitean orthogonal polynomials and their asymptotics are discussed but for a different problem.
\end{rem}
\medskip

\begin{rem} The fact that  the lower degree recurrence coefficients may not exist in the non-Hermitean case creates technical difficulties in deriving the usual string equations since, as was pointed out earlier,  this derivation requires that one be able to recursively relate higher degree recurrence coefficients all the way back to degree 0.  However, this issue poses no problems for the asymptotic difference string equations nor for the  asymptotic Toda equations which are what will be used in this paper. 
\end{rem} \medskip

\subsubsection{Path Weights and Recurrence Coefficients} In order to effectively utilize the relations (\ref{fl1}, \ref{fl2}) it will be essential to keep track of how the matrix entries of powers of the recurrence operator, $\mathcal{L}^j$, depend on the original recurrence coefficients. That is best done via the combinatorics of weighted walks on the index lattice of the orthogonal polynomials. The relevant walks here are {\it Motzkin paths} which are walks, $P$, on $\mathbb{Z}$ which, at each step, can increase by 1, decrease by 1 or stay the same. Set 
\begin{equation}  \label{motzkin}
\mathcal{P}^j(m_1, m_2) = \,\, \mbox{the set of all Motzkin paths of length $j$ from $m_1$ to $m_2$}.
\end{equation}
Then step weights, path weights and the $(m_1, m_2)$-entry of  $\mathcal{L}^j$ are, respectively, given by 
\begin{eqnarray} \nonumber
\omega(p) &=& \left\{ 
\begin{array}{cc}
 1 & \mbox{if the $p^{th}$ step moves from $n$ to $n+1$ on the lattice}\\
  a_n & \mbox{if the $p^{th}$ step stays at $n$}\\
 b^2_n & \mbox{if the $p^{th}$ step moves from $n$ to $n-1$}  
\end{array}
\right.\\
\nonumber \omega(P) &=& \prod_{\mbox{steps}\,\, p \in P} \omega(p)\\
\label{weights} \mathcal{L}^{j}_{m_1, m_2} &=& \sum_{P \in \mathcal{P}^j(m_1, m_2)} \omega(P).
\end{eqnarray}

\subsection{Motzkin Representation of the Difference String equations} \label{motzstring}

The {\em difference string equations}  are given (for the $j$-valent case) by  (\ref{fl1}):
\begin{equation}
\left[ \mathcal{L}, \left(\mathcal{L} + j t \mathcal{L}^{j-1} \right)_- \right] = \frac{1}{N} I \,.
\end{equation}
This leads to a pair of equations:

\begin{itemize}

\item
{\em the $(n+1, n)$ entry} gives
\begin{align} \label{string-star1}
0 &= ( a_{n+1} - a_n ) \left( \mathcal{L} + j t \mathcal{L}^{j-1} \right)_{n+1, n} + 
\left( \mathcal{L} + j t \mathcal{L}^{j-1} \right)_{n+2, n} \nonumber \\
& \hspace{2cm}
 - \left( \mathcal{L} + j t \mathcal{L}^{j-1} \right)_{n+1, n-1}\,, 
\end{align}

\item
and {\em the $(n, n)$ entry} gives
\begin{equation} \label{string-star2}
\frac{1}{N} = \left(\mathcal{L} + j t \mathcal{L}^{j-1} \right)_{n+1, n} - 
\left(\mathcal{L} + j t \mathcal{L}^{j-1} \right)_{n, n-1} \,.
\end{equation}

\end{itemize}
Let us work this out in terms of Motzkin paths for the particular case of $j=3$.  The equations for the diagonal and subdiagonal  equations reduce respectively  to 
\begin{eqnarray*}
\frac{x}{n} &=& \left(\mathcal{L}_{n+1,n} - \mathcal{L}_{n,n-1}\right) + 3 t  \left(\mathcal{L}^{2}_{n+1,n}  -  \mathcal{L}^{2}_{n, n-1} \right)\\
0 &=& (a_{n+1} - a_n)  \left(\mathcal{L}_{n+1,n} + 3 t \mathcal{L}^{2}_{n+1,n}\right) + \left(\mathcal{L}_{n+2,n} - \mathcal{L}_{n+1,n-1}\right) 
+ 3 t \left( \mathcal{L}^{2}_{n+2,n}  -  \mathcal{L}^{2}_{n+1, n-1} \right)
\end{eqnarray*}
where we have used the relation $x = \frac{n}N$.

Referring to (\ref{motzkin}), we see that the relevant path classes here are
\begin{eqnarray*}
\mathcal{P}^1(n+1, n) &=& \mbox{a descent by one step }\\
\mathcal{P}^1(n+2, n) &=& \mbox{the empty set}\\
\mathcal{P}^2(n+1, n) &=& \mbox{a horizontal step followed by a single descent} 
\\ & & \mbox{\qquad or a single descent followed by a horizontal step}\\
\mathcal{P}^2(n+2, n) &=& \mbox{two successive descent steps}.
\end{eqnarray*}
Note that the structure of the path classes does not actually depend upon $n$. This is a reflection of the underlying spatial homogeneity of these equations. Thus, for the purpose of describing the path classes, one can translate $n$ to $0$. 

Now applying (\ref{weights}) the difference string equations become
\begin{eqnarray*}
\frac1n &=& \left( b^2_{n+1} - b^2_n\right) + 3 t \left( a_{n+1} b^2_{n+1} + a_n b^2_{n+1} - a_n b^2_n - a_{n-1} b^2_n\right)\\
0 &=& (a_{n+1} - a_n) \left( b^2_{n+1} + 3 t  (a_{n+1} + a_n)  b^2_{n+1}\right) + 3 t \left( b^2_{n+2} b^2_{n+1} - b^2_{n+1} b^2_n\right)
\end{eqnarray*}
where, for this example, we have set the parameter $x$ equal to $1$. The coefficient $b^2_{n+1}$ is non-vanishing by (\ref{Hirota}) and the fact that the partition functions are non-vanishing for $n$ sufficiently large. Hence we may divide it out of the second equation to arrive at the slightly simpler system
\begin{eqnarray*}
\frac1n &=& \left( b^2_{n+1} - b^2_n\right) + 3 t \left(  b^2_{n+1}( a_{n+1} +  a_n) -  b^2_n (a_n + a_{n-1})\right)\\
0 &=& (a_{n+1} - a_n) \left( 1 + 3 t  (a_{n+1} + a_n) \right) + 3 t \left( b^2_{n+2} -  b^2_n\right).
\end{eqnarray*}

\subsection{Motzkin Representation of the Toda equations} \label{motztoda}

We now pass to a more explicit form of of the {\it Toda equations} (\ref{fl2}) in the case $j = 2\nu + 1$:
\begin{align}
-\frac{1}{N} \label{an} \frac{d a_n}{dt_{2\nu+1}} &= \left( \mathcal{L}^{2\nu+1}\right)_{n+1, n} -  \left( \mathcal{L}^{2\nu+1}\right)_{n, n-1}\\
\label{bn} 
-\frac{1}{N}  \frac{d b^2_n}{dt_{2\nu+1}} &= \left(a_n- a_{n-1}\right) \left( \mathcal{L}^{2\nu+1}\right)_{n, n-1} + 
\left( \mathcal{L}^{2\nu+1}\right)_{n+1, n-1} -  \left( \mathcal{L}^{2\nu+1}\right)_{n, n-2}.
\end{align}
To describe this in more detail we will once again specialize to the trivalent case ($\nu = 1$). There are two relevant path classes here:
\begin{eqnarray*}
\mathcal{P}^3(n+1, n) && \mbox{described in Figures \ref{2h} and \ref{0h} for the case $n=0$}\\
\mathcal{P}^3(n+2, n) && \mbox{described in Figure \ref{1h} for the case $n=0$}.
\end{eqnarray*} 
The latter case corresponds to what was used in \cite{EMP08} but for Dyck paths (Motzkin paths without any horizontal steps) of length $2\nu$.

\begin{figure}[h] 
\begin{center}
\resizebox{5in}{!}{\includegraphics{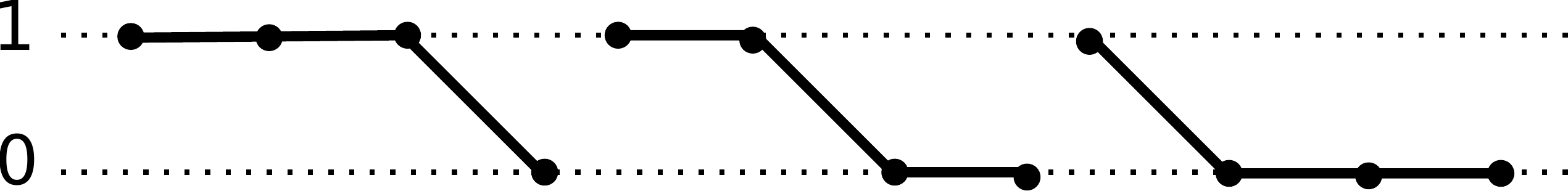}}
\end{center}
\caption{\label{2h} Elements of  $\mathcal{P}^3(1, 0)$ with two horizontal steps}
\end{figure}

\begin{figure}[h] 
\begin{center}
\resizebox{5in}{!}{\includegraphics{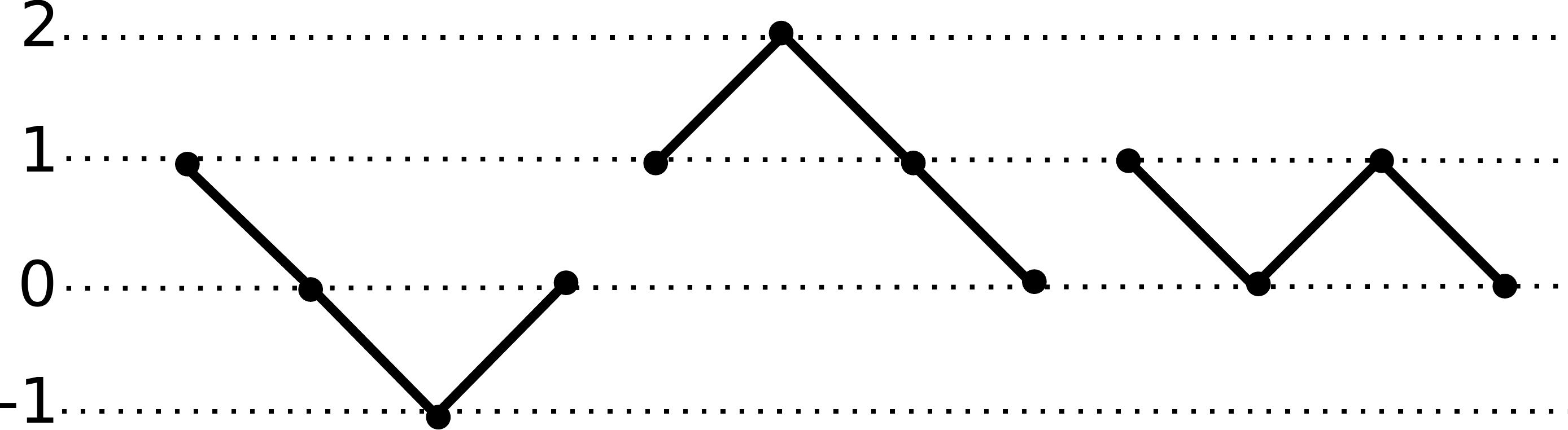}}
\end{center}
\caption{\label{0h} Elements of $\mathcal{P}^3(1, 0)$ with no horizontal steps (Dyck paths)}
\end{figure}

Applying (\ref{weights}),  the trivalent Toda equations become
\begin{align*}
- \frac1n \frac{da_n}{dt} &= \left(a^2_{n+1}b^2_{n+1} - a^2_n b^2_n\right) + \left(a_{n+1} a_n b^2_{n+1} - a_n a_{n-1} b^2_n\right) 
+ \left(a^2_{n}b^2_{n+1} - a^2_{n-1} b^2_n\right)\\
&\hspace{1cm} +  \left(b^2_{n+1} b^2_n - b^2_n b^2_{n-1}\right) +  \left(b^2_{n+2} b^2_{n+1} - b^2_{n+1} b^2_{n}\right) + 
\left(b^2_{n+1} b^2_{n+1} - b^2_n b^2_{n}\right)\\
- \frac1n \frac{db^2_n}{dt} &=  \left( a_{n} - a_{n-1} \right) \left[ a^2_{n}b^2_{n} + a_{n} a_{n-1} b^2_{n} + a^2_{n-1}b^2_{n} + b^2_{n} b^2_{n-1} + b^2_{n+1} b^2_{n} + b^2_{n} b^2_{n} \right]\\
&\hspace{1cm} + \left(a_{n+1} b^2_{n+1} b^2_{n} - a_{n} b^2_{n} b^2_{n-1}\right) 
+ \left(a_{n} b^2_{n+1} b^2_{n} - a_{n-1} b^2_{n} b^2_{n-1}\right) \\
&\hspace{1cm} + \left(a_{n-1} b^2_{n+1} b^2_{n} - a_{n-2} b^2_{n} b^2_{n-1}\right) 
\end{align*}
where we have again used the relation $x = \frac{n}N$ and then set the parameter $x = 1$.

\begin{figure}[h] 
\begin{center}
\resizebox{5in}{!}{\includegraphics{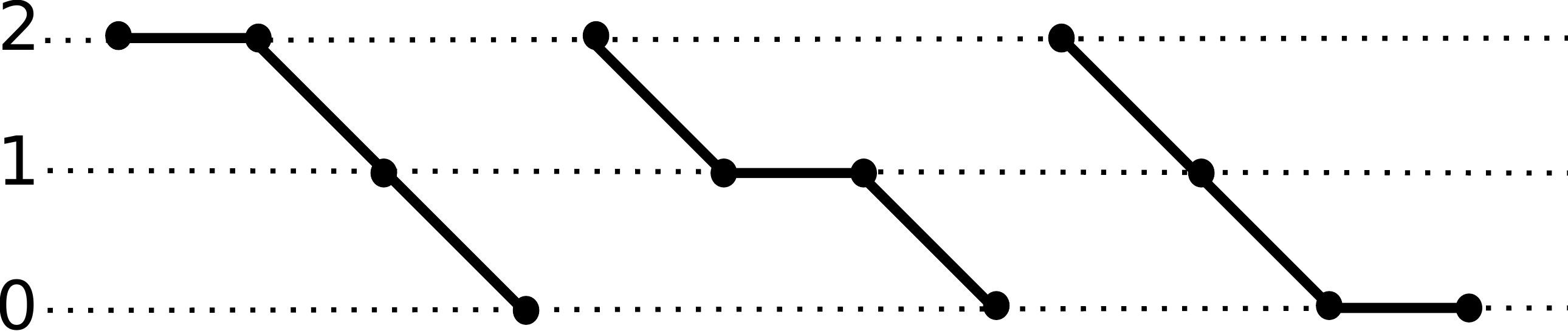}}
\end{center}
\caption{ \label{1h} Elements of $\mathcal{P}^3(2, 0)$}
\end{figure}

\section{Continuum Limits} \label{sec:44}
A number of discrete variables will appear in the following discussion as we prepare to make the transition to the continuum limit. Not all of these discrete variables will be directly involved in the description of that limit. However, in order to avoid confusion, it is perhaps best that we start by briefly describing all of these variables and their interrelations as well as their relations with other variables and parameters. As indicated at the outset, the positive parameter $N$ sets the scale for the potential in the random matrix partition function and, throughout this paper, it is taken to be large. The discrete variable $n$ labels the lattice {\it position} on $\mathbb{Z}^{\geq 0}$ that marks, for instance, the $n^{th}$ orthogonal polynomial and recurrence coefficients, diagonal or sub-diagonal.  We also always take $n$ to be large and in fact to be of the same order as $N$. As stated in the Introduction, it is to be understood that as $n$ and $N$ tend to $\infty$, they do so in such a way that their ratio
\begin{equation}\label{xdef} 
x \doteq \frac nN
\end{equation}
remains fixed at a value close to $1$. In fact within all subsequent proofs and derivations $x$ itself will be held fixed.

In addition to the {\it global} or {\it absolute} lattice variable $n$, we also introduce a {\it local} or {\it relative} lattice variable which we will denote by $k$. It varies over integers but will always be taken to be small in comparison to $n$ and independent of $n$. We will frequently study expressions involving $n + k$ which we will think of as small discrete variations around a large value of  $n$. We have already encountered this in the form of the difference string or Toda equations. The spatial homogeneity of those equations manifests itself in their all having the same form, independent of what $n$ is, while $k$ in those equations varies over $\{-\nu - 1, \dots, -1,0,1, \dots, \nu + 1\}$, as explicitly displayed for the trivalent case ($\nu = 1$) in subsections \ref{motzstring} and \ref{motztoda}.  Indeed in what follows it will suffice to take $\nu + 1 << n$ in order to insure the necessary separation of scales between $k$ and $n$.

We introduce a number of other scalings of variables that we will be using.
\begin{eqnarray}
\label{nuscaling1} s_1 &\doteq& x^{-\frac12} t_1
\end{eqnarray}
or more generally for $\nu \in \mathbb{Z}^+$,
\begin{eqnarray}
\label{nuscaling2} s_{2\nu + 1} &\doteq& x^{\nu - \frac12} t_{2\nu+1}\\
\label{wdef}\tilde{w} &\doteq& 1 + \frac kn.
\end{eqnarray} 
The variable $\tilde{w}$ is introduced for two reasons. First, to help make some of the subsequent expressions less cumbersome. Which value of $k$ is intended will be clear from the context or it will be made explicit. The second reason for introducing this variable is that it provides the transition to the continuum equations. As we shall see, at a certain point in the subsequent arguments $\tilde{w}$ will appear within the arguments of coefficient functions of the large $n$ asymptotic expansions. Since these coefficient functions are in fact analytic in their arguments, we will take advantage of this fact to regard these functions as analytic functions of 
$\tilde{w}$ regarded as a continuous variable.

In Theorem \ref{priors} we will present a more precise statement and extension of prior results on the free energy expansion (\ref{I.002}) as they will relate to what we do in the remainder of this paper. However, before getting to that we need to recall a preliminary prior result: 
\begin{prop}\label{tauscale} \cite{EMP08}
\begin{eqnarray*}
\tau^2_{n+k,N}( t_1, t_{2\nu+1}) &=& \tau^2_{n+k,n+k} \left( \left( \frac{n+k}{N}\right)^{-1/2} t_1,  \left( \frac{n+k}{N}\right)^{\nu - 1/2} t_{2\nu+1}\right)\\
&=& \tau^2_{n+k,n+k}  \left(\left( 1 + \frac{k}n\right)^{-1/2} s_1,  \left( 1 + \frac{k}n\right)^{\nu - 1/2} s_{2\nu+1}\right)\\
&=& \tau^2_{n+k,n+k}  \left(\tilde{w}^{-1/2} s_1,  \tilde{w}^{\nu - 1/2} s_{2\nu+1}\right).
\end{eqnarray*}
\end{prop}
\begin{proof}
The first equality follows from an appropriate change of variables in the Szeg\"o representation (\ref{szego}, \ref{szego1}), the second just applies the definitions (\ref{nuscaling1}) and (\ref{nuscaling2}) and the third applies (\ref{wdef}). The change of variables in the Szeg\"o representation amounts to introducing the re-scaling
${\lambda}_j = \sqrt{x}\hat{\lambda}_j$, from which we then have
\begin{eqnarray}  \nonumber
 \hat{Z}_N^{(n)}({t_1}, {t_{2\nu+1}}) &=&
x^{n^2/2}\int \cdots \int
  \exp\bigg(-n \bigg\{\sum_{j=1}^{n} \bigg(\frac{1}{2}\hat{\lambda}_j^2
+ {t}_{2\nu+1} x^{\nu-1/2}\hat{\lambda}_j^{2\nu+1}
 \\ \nonumber 
&& \hspace{5cm}
 + \frac{{t_1}}{\sqrt{x}}
\hat{\lambda}_j\bigg) \bigg\}\bigg)
  {\mathcal{V}} (\hat{\lambda}) \,\, d^{n} \hat{\lambda}\\
&=& x^{n^2/2} \hat{Z}_n^{(n)}\left(s_1, s_{2\nu+1}
\right),
\end{eqnarray}
where $ {\mathcal{V}} ({\lambda}) = \prod_{j<\ell}\left| \lambda_j  - \lambda_\ell \right|^2$. Shifting from $n$ to $n+k$ this becomes
\begin{eqnarray*}
\hat{Z}_N^{(n+k)}({t_1}, {t_{2\nu+1}}) &=& x^{(n+k)^2/2} \hat{Z}_{n+k}^{(n+k)}\left(\left( 1 + \frac{k}n\right)^{-1/2}s_1, \left( 1 + \frac{k}n\right)^{\nu - 1/2} s_{2\nu+1}\right)
\end{eqnarray*}
The proposition follows immediately from this.
\end{proof}
\medskip

\begin{rem} Starting in subsection \ref{ssec:dse} we will in general be setting $x=1$. That is because the main focus in this paper is on the structure of the partition function coefficients as generating functions for map enumeration. However, for (possible, future) applications to the statistics of random matrix eigenvalues, the ability to asymptotically ``detune" the matrix size away from the scale of the potential is important. Therefore, we have chosen to keep this parameter free up through this preliminary subsection.
\end{rem}\smallskip

We introduce one more notational definition:
\begin{equation} \label{Delta}
\Delta_k \log \tau^2_{n,n}(s_1, s_{2\nu+1}) \doteq \log \tau^2_{n+k,n+k}(\tilde{w}^{-1/2} s_1,  \tilde{w}^{\nu - 1/2} s_{2\nu+1}) - \log \tau^2_{n,n}(s_1, s_{2\nu+1}).
\end{equation}

\begin{thm} \cite{EM03, EMP08, BD10} \label{priors}
\begin{align} \nonumber
\log \tau^2_{n+k,n+k}\left(s_1, s_{2\nu+1}\right) &=  n^2 \tilde{w}^2  e_0(\tilde{w}^{-1/2}{s}_1, \tilde{w}^{\nu - 1/2}{s}_{2\nu+1})  +  \\ \nonumber
&\phantom{=} \qquad\qquad
+   e_1(\tilde{w}^{-1/2}{s}_1, \tilde{w}^{\nu - 1/2} s_{2\nu+1}) +  \dots  \\ \label{I.003}
 &\phantom{=} \quad
+  \frac{1}{n^{2g-2}} \tilde{w}^{2-2g} e_g(\tilde{w}^{-1/2}{s}_1, \tilde{w}^{\nu - 1/2}{s}_{2\nu+1}) + \dots 
\end{align}
is an asymptotic expansion in $n^{-2}$, uniformly valid for $(s_1, s_{2\nu+1}) \in \mathcal{K} = $ any compact subset of $(-\delta, \delta) \times [0, s^{(\nu, g)}_c)$ and $|k| \leq \nu + 1$ with $\frac{\nu + 1}{n} < \epsilon$ for $\epsilon$ and $\delta$ sufficiently small. 
(Here $s^{(\nu,g)}_c$ is a fixed positive constant depending only on $\nu$ and $g$.) Explicitly, this means that for each $g$ there is a constant, $C_g$, depending only on $\nu$ and  
$\mathcal{K}$ such that 
\begin{align*}
\left| \log \tau^2_{n+k,n+k}\left(s_1, s_{2\nu+1}\right) - n^2 \tilde{w}^2  e_0(\tilde{w}^{-1/2}{s}_1, \tilde{w}^{\nu - 1/2}{s}_{2\nu+1})  -  \dots \right. & \\ \left. 
 - \frac{1}{n^{2g-2}} \tilde{w}^{2-2g} e_g(\tilde{w}^{-1/2}{s}_1, \tilde{w}^{\nu - 1/2}{s}_{2\nu+1}) \right|  &\leq \frac{C_g}{n^{2g}}
\end{align*}
for all $(s_1, s_{2\nu+1}) \in \mathcal{K}$ and $|k| \leq \nu + 1$. 
\begin{itemize}
\item[a)] This expansion may be differentiated term by term in $s_1, s_{2\nu+1}$ with the same type of uniformity except that the constant $C_g$ will now also depend on the multi-index of the derivatives.
Moreover,  the coefficient $e_g(w^{-1/2}{s}_1, w^{\nu - 1/2}{s}_{2\nu+1})$ and its mixed derivatives in $(w, s_1, s_{2\nu+1})$, after being evaluated at $(w, s_1) = (1,0)$, are all complex analytic in a disc of radius $s^{(\nu, g)}_c$ centered at $0$ in the complex $s_{2\nu+1}$ plane. (Here we introduce $w$ as a continuous complex variable replacing $\tilde{w}$; hence, one may differentiate $e_g$ with respect to it.) One expects this radius of convergence to be independent of $g$ as is known to be true in the case of even weights \cite{Er09}.
\item[b)] One also has an asymptotic expansion for differences $\Delta_k \log \tau^2_{n,n}$:
\begin{eqnarray} \label{I.004}
\Delta_k \log \tau^2_{n,n}\left(s_1, s_{2\nu+1}\right) &=&  \sum_{g=0}^{\infty} \frac{1}{n^{2g-2}}  
\\ && \nonumber
\sum_{j=1}^\infty \frac1{j!} \frac{\partial^j}{\partial w^j} w^{2-2g} e_g( w^{-1/2} s_1, w^{\nu - 1/2} s_{2\nu+1} )|_{w=1}  \left(  \frac {k} n \right)^j
\end{eqnarray}
where we regard $w$ as a continuous complex variable. Once again this expansion is uniformly valid for $(s_1, s_{2\nu+1}) \in \mathcal{K}$ and $|k| \leq \nu + 1$, by which we mean that  for each $g$ there is a constant, $D_g$, depending only on $\nu$ and $\mathcal{K}$  such that 
\begin{align}
\label{Delest}
 \left| \Delta_k \log \tau^2_{n,n}\left(s_1, s_{2\nu+1}\right) 
- \sum_{m=0}^{g} \frac{1}{n^{2m-2}}
\right. \hspace{5cm} & \\ \nonumber 
 \left. 
 \sum_{j=1}^{2(g-m)+1} \frac1{j!} \frac{\partial^j}{\partial w^j} w^{2-2m} e_m\left(w^{-1/2} s_1, w^{\nu - 1/2} s_{2\nu+1} \right)|_{w=1} \left( \frac kn\right)^j\right| &\leq \frac{D_g}{n^{2g}}. 
\end{align}
This expansion may be differentiated term by term in $s_1, s_{2\nu+1}$ preserving uniformity. 
\end{itemize}
\end{thm} 
\begin{proof}
The basic result is that of \cite{EM03} extended, in \cite{BD10}, to the case of weights with odd dominant power  (see (\ref{I.002}) \ref{unif} - \ref{diff}). From this it follows that one has constants 
$\hat{C}_g$ depending only on $\nu, \mathcal{K}$ such that 
 \begin{eqnarray*}
 \left| \log \tau^2_{n+k,n+k}\left(s_1, s_{2\nu+1}\right) - (n+k)^2   e_0(\tilde{w}^{-1/2}{s}_1, \tilde{w}^{\nu - 1/2}{s}_{2\nu+1})  -  \dots  \right. \\
\left. \dots 
 - \frac{1}{(n+k)^{2g-2}}  e_g(\tilde{w}^{-1/2}{s}_1, \tilde{w}^{\nu - 1/2}{s}_{2\nu+1}) \right|  &\leq & \frac{\hat{C}_g}{(n+k)^{2g}}
 \end{eqnarray*}
 We then rewrite the above equation using $n+k = n\tilde{w}$,
 \begin{eqnarray*}
  \left| \log \tau^2_{n+k,n+k}\left(s_1, s_{2\nu+1}\right) - n^2 \tilde{w}^2   e_0(\tilde{w}^{-1/2}{s}_1, \tilde{w}^{\nu - 1/2}{s}_{2\nu+1})  -  \dots \right. \\ \left. \dots  
 -  \frac{1}{n^{2g-2}} \tilde{w}^{2-2g} e_g(\tilde{w}^{-1/2}{s}_1, \tilde{w}^{\nu - 1/2}{s}_{2\nu+1}) \right|  &\leq & \frac{\hat{C}_g}{n^{2g}\tilde{w}^{2g}}.
\end{eqnarray*}
The desired estimate is realized by taking $C_g = \frac{\hat{C}_g}{(1 - \epsilon)^{2g}} \geq \frac{\hat{C}_g}{\tilde{w}^{2g}}$.

The ensuing statements of the theorem, in (a), also follow directly from these prior results. In particular, the analyticity of $e_g$ in its arguments follows from (\ref{I.002}) \ref{analyt}; mixed derivatives then yield linear combinations of derivatives of $e_g$ with respect to its arguments whose coefficients are polynomials in $s_1, s_{2\nu+1}$ and fractional powers of $w$ (which is bounded away from zero). Evaluating at  $(w, s_1) = (1,0)$ then yields a linear combination of derivatives of $e_g$ with coefficients that are polynomial is $s_{2\nu+1}$. By (\ref{analyt}) this linear combination is analytic in a disc as stated in the theorem.

For (b) observe that a straightforward estimate of the difference of the asymptotic expansions for $\log \tau^2_{n+k,n+k}$ and $\log \tau^2_{n,n}$ yields
\begin{eqnarray*}
&& \left| \log \tau^2_{n+k,n+k}\left(s_1, s_{2\nu+1}\right) - \log \tau^2_{n,n} \left(s_1, s_{2\nu+1}\right) - 
n^2 w^2  e_0(w^{-1/2}s_1, w^{\nu - 1/2}s_{2\nu+1}) \right. \\
&& \hspace{2cm} +    e_0(s_1, s_{2\nu+1})  -  \dots 
\\ && \hspace{1cm} \left.
- \frac{1}{n^{2g-2}} w^{2-2g} e_g(w^{-1/2}{s}_1, w^{\nu - 1/2}{s}_{2\nu+1}) +  e_g(s_1, s_{2\nu+1})\right|_{w=1 + \frac kn}  \leq \frac{2 C_g}{n^{2g}}\,.
\end{eqnarray*}
We rewrite this as 
\begin{align*}
 \bigg| \Delta_k \log \tau^2_{n,n} \left(s_1, s_{2\nu+1}\right) 
\hspace{8cm} & \\    
-  \sum_{m=0}^g \frac{1}{n^{2m-2}} \left(w^{2-2m} e_m(w^{-1/2}{s}_1, w^{\nu - 1/2}{s}_{2\nu+1}) -  e_g(s_1, s_{2\nu+1})\right)\bigg|_{w=1 + \frac kn}  &\leq \frac{2 C_g}{n^{2g}}\,,
\end{align*}
which is
\begin{align*}
  \bigg| \Delta_k \log \tau^2_{n,n} \left(s_1, s_{2\nu+1}\right)  \hspace{8cm} &
\\     - \sum_{m=0}^g \frac{1}{n^{2m-2}} \bigg((1+\frac kn)^{2-2m} e_m((1+\frac kn)^{-1/2}{s}_1, (1+\frac kn)^{\nu - 1/2}{s}_{2\nu+1}) & 
 \\ 
-  e_g(s_1, s_{2\nu+1})\bigg)\bigg|  &\leq \frac{2 C_g}{n^{2g}}\,.
\end{align*}
One now Taylor expands the $e_m$ terms centered at $w=1$ and evaluated at $1 + \frac kn$. For notational convenience set
\begin{equation*}
F_m\left(w,s_1, s_{2\nu+1}\right) = w^{2-2m} e_m(w^{-1/2}{s}_1, w^{\nu - 1/2}{s}_{2\nu+1}).
\end{equation*}
Then this expansion has the form
\begin{align*}
 \left| \Delta_k \log \tau^2_{n,n}\left(s_1, s_{2\nu+1}\right) - \sum_{m=0}^{g} \frac{1}{n^{2m-2}} \sum_{j=1}^{2(g-m)+1} \frac1{j!} \frac{\partial^j}{\partial w^j} F_m\left(w,s_1, s_{2\nu+1}\right)|_{w=1} \left( \frac kn\right)^j 
\right. \\ \left.
+ R^{(m)}\left(w,s_1, s_{2\nu+1}\right) \left(\frac kn\right)^{2(g-m+1)}\right| \leq \frac{2 C_g}{n^{2g}} 
\end{align*}
where $R^{(m)}\left(w,s_1, s_{2\nu+1}\right)$ denotes the remainder term, of order $2(g-m+1)$ in $w$, for $F_m$. By elementary inequalities one then has
\begin{eqnarray*}
&& \bigg| \Delta_k \log \tau^2_{n,n}\left(s_1, s_{2\nu+1}\right)  
\\
&&
- \sum_{m=0}^{g} \frac{1}{n^{2m-2}} \sum_{j=1}^{2(g-m)+1} \frac1{j!} \frac{\partial^j}{\partial w^j} w^{2-2m} e_m\left(w^{-1/2} s_1, w^{\nu - 1/2} s_{2\nu+1} \right)|_{w=1} \left( \frac kn\right)^j\bigg|\\
& \leq& \frac{2 C_g}{n^{2g}} + \frac1{n^{2g}}\sum_{m=0}^g \left| R^{(m)}\left(w,s_1, s_{2\nu+1}\right)\right| |k|^{2(g-m+1)}. 
\end{eqnarray*} 
By Cauchy's remainder theorem for analytic functions one has
\begin{eqnarray*}
\left| R^{(m)}\left(w,s_1, s_{2\nu+1}\right)\right| &\leq& 2^{2(g-m+1)} d_m\\
d_m &=& \max_{(s_1, s_{2\nu+1}) \in \mathcal{K}} \max_{|w-1| = 1/2} F_m\left(w,s_1, s_{2\nu+1}\right). 
\end{eqnarray*}
(We assume here, as we may, that $\epsilon < 1/2$.)
Thus statement (b) is established by taking
\begin{eqnarray*}
D_g &=& 2 C_g + \sum_{m=0}^g d_m |2\nu + 2|^{2(g-m+1)}.
\end{eqnarray*}
\end{proof}

\begin{rem}  We mention here that the variables $s_j$ as defined above differ slightly from their usage in related works \cite{EMP08, Er09} where $s_j = -c_j t_j$ for appropriate constants $c_j > 0$. Also for comparison with \cite{BD10}, $s_3 = - u$ where $u$ is the weight parameter in that work. We further observe that because of the combinatorial interpretation of these generating functions, one can show that $e_g(s_{2\nu+1})$ is even in $s_{2\nu+1}$. Thus odd derivatives of $e_g$ are odd functions and even derivatives are even. 
\end{rem} \medskip

\begin{rem} By collecting terms in (\ref{Delest}) of the same order in $n^{-1}$ one sees that the asymptotic series represented by (\ref{I.004}) does indeed have a uniformly valid, well ordered expansion in inverse powers of $n$. The precise form of the coefficients in this {\it re-summed} expansion is not prima facie obvious; however, the results in the remainder of this paper will show precisely how these coefficients and those of other similarly derived asymptotic series  can in fact be determined. A key point for this process may already be observed in (\ref{Delest}); namely, all terms in the re-summation will be in the form  of differential expressions in the continuous variable $w$ which are then uniformly evaluated at $w=1$. In the relevant settings, the coefficients of the inverse powers of $n$ may be regarded as hierarchies of differential equations in $w$ which are to be solved and whose solutions are then evaluated at $w=1$ in order to yield explicit expressions for generating functions (in $s_{2\nu+1}$) and similar functions of combinatorial or statistical interest. In the remainder of this section we build upon these ideas to derive the general form of the various hierarchies of differential equations. In sections \ref{sec:55} and \ref{sec:4} we carry out this process in complete detail for the trivalent case ($\nu = 1$). To get an idea of how the whole strategy comes together the reader might find it useful to browse these last two sections before proceeding systematically through the general derivations which begin in subsection \ref{ssec:dse}.   
\end{rem}\medskip

We will make essential use of the Hirota formulas for the Toda variables in their original scaling.
\begin{lem} (Hirota)
\begin{align}
\label{a} a_{n, N} &= -\frac{1}{N} \frac{\partial}{\partial t_1} \log \left[ \frac{\tau^2_{n+1, N}}{\tau^2_{n, N}} \right]  
                                  =  -\frac{1}{N} \frac{\partial}{\partial t_1} \log \left[ \frac{Z_N^{(n+1)}(t_1, t_{2\nu+1})}{Z_N^{(n)}(t_1, t_{2\nu+1})} \right]\\
\label{b}  b_{n, N}^2 &= \frac{1}{N^2} \frac{\partial^2}{\partial t_1^2} \log \tau^2_{n, N}
                                       = \frac{1}{N^2} \frac{\partial^2}{\partial t_1^2} \log \frac{1}{N^2} Z_N^{(n)}(t_1, t_{2\nu+1})\,,
\end{align}
where the factors of $1/N$ are consistent with the energy scaling chosen in the definition of $\mu$ (\ref{RMT}).
\end{lem}
\begin{proof}
 From (\ref{an}) and (\ref{bn}) one deduces that the Toda equations for $\nu = 0$ are
\begin{eqnarray*}
- \frac1N \frac{da_{n,N}}{dt_1} &=& b^2_{n+1,N} - b^2_{n,N}\\
- \frac1N \frac{db^2_{n,N}}{dt_1} &=& b^2_{n,N} \left(a_{n,N} - a_{n-1, N}\right).
\end{eqnarray*}
Substituting the fundamental identity (\ref{Hirota}) into the second of these equations one has
\begin{eqnarray*}
a_{n,N}  - a_{n-1, N} &=& - \frac1N \frac{d}{dt_1} \left( \log \tau^2_{n+1, N}   - 2 \log \tau^2_{n, N} - \log \tau^2_{n-1, N}\right).
\end{eqnarray*}
From the Szeg\"o representation one has that the $\tau^2_{n, N}$ are simultaneously analytic in $(t_1, t_{2\nu+1})$.
By continuation of $(t_1, t_{2\nu+1})$ back to $(0, 0)$, the recurrence coefficients become those of the Hermite polynomials. From these initial values and the previous line one may deduce that, in fact, 
\begin{eqnarray*}
a_{n,N} &=& - \frac1N \frac{d}{dt_1}  \log \frac{\tau^2_{n+1, N}}{\tau^2_{n, N}}
\end{eqnarray*}
which is the first Hirota relation. Substituting this into the first Toda equation above one may similarly derive the second Hirota relation. 
The expression in terms of the partition function in each case follows directly from (\ref{tausquare}).
\end{proof}
\begin{cor} \label{thm:rec-asymp}
\begin{eqnarray}
\nonumber a_{n+k, N}(t_1, t_{2\nu+1})  & = &  -\frac{1}{N} \frac{\partial}{\partial t_1} \left[ \log \tau^2_{n+k+1,N}(t_1, t_{2\nu+1}) - 
\log \tau^2_{n+k,N}( t_1, t_{2\nu+1}) \right]  
\\
\label{anN-preexpansion} &=&  -\frac{x^{1/2} }{n} \frac{\partial}{\partial s_1} \Delta_{1} \log \tau^2_{n+k,n+k}\left(s_1, s_{2\nu+1}\right)\\
\nonumber b_{n+k, N}^2(t_1, t_{2\nu+1}) &=& \frac{1}{N^2} \frac{\partial^2}{\partial t_1^2} \log \tau^2_{n+k, N}\\
\label{bnN-preexpansion}  &=& \frac{x}{n^2} \frac{\partial^2}{\partial s_1^2} \left[\log \tau^2_{n,n}(s_1, s_{2\nu+1}) + \Delta_k \log \tau^2_{n,n}(s_1, s_{2\nu+1})\right].
\end{eqnarray}
\end{cor}
\begin{proof}
These representations follow directly from the Hirota relations (\ref{a}), (\ref{b}), along with (\ref{nuscaling1}), (\ref{nuscaling2}) and the definition (\ref{Delta}) of $\Delta_k$.
\end{proof}
\medskip

We have the following asymptotic expansions for  $a_{n+k,N}$, $b_{n+k,N}$.  Moreover, expanding the $\partial/\partial w$ and $\partial/\partial s_1$ derivatives within the coefficients in these asymptotic expansions one can see that these coefficients acquire a self-similar scaling.
\begin{thm} \label{thm:hf} The following are asymptotic series in $1/n$:
\begin{align} 
\label{h1}&\\
\nonumber  a_{n+k, N} &= h(s_1, s_{2\nu+1},\tilde{w}) =  x^{1/2} \sum_{g \geq 0} h_g(s_1, s_{2\nu+1},\tilde{w}) n^{-g}\\
\nonumber  h_g(s_1, s_{2\nu+1}, \tilde{w}) &= -  \tilde{w}^{1-g} \times
\\ \nonumber & \hspace{-0.5cm}
\sum_{\begin{matrix} 2 g_1 + j = g+1 \\  g_1 \geq 0\,, j>0\end{matrix}}  \frac1{j!} \frac{\partial^{j+1}}{\partial s_1 \partial w^j} \left[ w^{2-2g_1} e_{g_1}\left((w\tilde{w})^{-1/2} s_1, (w\tilde{w})^{\nu-1/2} s_{2\nu+1}\right)\right]_{w = 1}\\
 \label{f1} &\\
\nonumber b^2_{n+k,N} &= f(s_1, s_{2\nu+1},\tilde{w}) = x \sum_{g \geq 0} f_g(s_1, s_{2\nu+1},\tilde{w})  n^{-2g}\\
\nonumber f_g(s_1, s_{2\nu+1}, \tilde{w}) &= \tilde{w}^{2-2g} \frac{\partial^2}{\partial s_1^2}  e_g(\tilde{w}^{-1/2} s_1, \tilde{w}^{\nu-1/2} s_{2\nu+1}). 
\end{align}
Moreover
\begin{eqnarray}
\label{hw}h_g(s_1, s_{2\nu+1},w) &=& w^{\frac{1}{2} -g}  u_g(s_1 w^{-1/2}, s_{2\nu+1} w^{\nu - \frac{1}{2}})     \\
\label{fw} f_g(s_1, s_{2\nu+1},w)  &=& w^{1 -2g}  z_g(s_1 w^{-1/2}, s_{2\nu+1} w^{\nu - \frac{1}{2}})  
 \end{eqnarray}
 where $u_g$ and $z_g$ are analytic functions of their arguments in a neighborhood of $(0,0)$ and $w$ is a continuous variable in terms of which the general form of these coefficient functions is described..
\end{thm}
\begin{proof}
We first consider (\ref{f1}).  By (\ref{bnN-preexpansion}), (\ref{I.003}) and the fact that these asymptotic series may be differentiated term by term, one has
\begin{align}
\nonumber b^2_{n+k,N} = f(s_1, s_{2\nu+1},\tilde{w}) &= x \sum_{g=0}^{\infty} \frac{1}{n^{2g}} \tilde{w}^{2-2g}  \frac{\partial^2}{\partial s_1^2}e_g( \tilde{w}^{-1/2} s_1, \tilde{w}^{\nu - 1/2} s_{2\nu+1} ) \\
\label{bnN-expansion}  &= x \sum_{g=0}^{\infty} \frac{1}{n^{2g}} \tilde{w}^{1-2g}  \frac{\partial^2}{\partial q^2}e_g( q, \tilde{w}^{\nu - 1/2} s_{2\nu+1} )|_{q=\tilde{w}^{-1/2} s_1}.
\end{align}
Thus we define 
\begin{equation*}
z_g(s_1 w^{-1/2}, s_{2\nu+1} w^{\nu - \frac{1}{2}}) \doteq  \frac{\partial^2}{\partial q^2}e_g( q, w^{\nu - 1/2} s_{2\nu+1} )|_{q= w^{-1/2} s_1}
\end{equation*}
which, in light of Theorem \ref{priors}, establishes all claims concerning $f$ and $f_g$.  The case for $h$ and $h_g$ proceeds in essentially the same manner but is a bit more complicated. By (\ref{anN-preexpansion}) and (\ref{I.003}) one has 
\begin{eqnarray*}
a_{n+k,N} &=& -\frac{x^{1/2} }{n} \frac{\partial}{\partial s_1} \sum_{g \geq 0} \left[ (n+k+1)^{2-2g} e_g\left( (\tilde{w} + 1/n)^{-1/2} s_1, (\tilde{w} + 1/n)^{\nu-1/2} s_{2\nu+1} \right) \right.
\\ && \phantom{-\frac{x^{1/2} }{n} \frac{\partial}{\partial s_1} \sum_{g \geq 0}}
\left. - (n+k)^{2-2g} e_g\left( \tilde{w}^{-1/2} s_1, \tilde{w}^{\nu-1/2} s_{2\nu+1} \right) \right]. 
\end{eqnarray*} 
Setting $\hat{w} = 1 + \frac{1}{n+k}$ this may be rewritten as
\begin{eqnarray*}
 h(s_1, s_{2\nu+1},\tilde{w}) &=&  -\frac{x^{1/2}  }{n} \frac{\partial}{\partial s_1}\sum_{g \geq 0} (n+k)^{2-2g} \left[ \hat{w} ^{2-2g} e_g\left((\hat{w} \tilde{w})^{-1/2} s_1, (\hat{w} \tilde{w})^{\nu-1/2} s_{2\nu+1}\right) 
\right. \\ && \phantom{-\frac{x^{1/2}  }{n} \frac{\partial}{\partial s_1}\sum_{g \geq 0}} \left.
- e_g\left( \tilde{w}^{-1/2} s_1, \tilde{w}^{\nu-1/2} s_{2\nu+1} \right)\right].
 \end{eqnarray*}
We next expand the summands in terms of Taylor series in the continuous variable $w$ centered at $w=1$ and evaluated at 
$w = 1 + \frac{1}{n+k}$: 
\begin{align*} 
h(s_1, s_{2\nu+1},\tilde{w}) &= -\frac{x^{1/2}  }{n} \frac{\partial}{\partial s_1} \sum_{g \geq 0} (n+k)^{2-2g} \times
\\ &
\sum_{j\geq 1} \frac1{j!} \frac{\partial^j}{\partial w^j} \left[ w^{2-2g} e_g\left((w \tilde{w})^{-1/2} s_1, (w \tilde{w})^{\nu-1/2} s_{2\nu+1}\right)\right]_{w = 1} \frac{1}{(n+k)^j}\\
&\hspace{-0.25in}
= -\frac{x^{1/2} }{n} \frac{\partial}{\partial s_1} \sum_{g \geq 0} (n+k)^{1-g} 
\times \\ &
\sum_{\begin{matrix} 2 g_1 + j = g+1 \\  g_1 \geq 0\,, j>0\end{matrix}}  \frac1{j!} \frac{\partial^j}{\partial w^j} \left[  w^{2-2g_1} e_{g_1}\left((w \tilde{w})^{-1/2} s_1, (w\tilde{w})^{\nu-1/2} s_{2\nu+1}\right)\right]_{w = 1} \\
&\hspace{-0.25in} 
= x^{1/2}  \sum_{g \geq 0} n^{-g} \tilde{w}^{1-g} 
\sum_{\begin{matrix} 2 g_1 + j = g+1 \\  g_1 \geq 0\,, j>0\end{matrix}}  - \frac1{j!} \frac{\partial^{j+1}}{\partial s_1\partial w^j} \bigg[ 
\\ & \hspace{3cm} 
w^{2-2g_1} e_{g_1}\left((w\tilde{w})^{-1/2} s_1, (w\tilde{w})^{\nu-1/2} s_{2\nu+1}\right)\bigg]_{w = 1}. 
\end{align*}
In the second equality above, we have collected the coefficients of $(n+k)^{1-g}$; and in the last equality, we make use of the relation $n+k =  n\tilde{w}$. 
To see the self-similar structure of the internal sum in this last line, observe that
\begin{align*}
\frac{\partial}{\partial w} w^m E\left( (w\tilde{w})^{-1/2} s_1, (w\tilde{w})^{\nu-1/2} s_{2\nu+1} \right) &= m w^{m-1} E\left((w\tilde{w})^{-1/2} s_1, (w\tilde{w})^{\nu-1/2} s_{2\nu+1} \right) \\
& \hspace{-2in} 
+ (\nu-\frac12) w^{m -1} \left( (w\tilde{w})^{\nu-1/2} s_{2\nu+1} \right) \frac{\partial}{\partial q_2} E(w\tilde{w})^{-1/2} s_1, q_2 )_{q_2 = (w\tilde{w})^{\nu-1/2} s_{2\nu+1}} \\
&\hspace{-1.75in} 
- \frac12 w^{m -1} \left( (w\tilde{w})^{-1/2} s_{1} \right) \frac{\partial}{\partial q_1} E(q_1, (w\tilde{w})^{\nu-1/2} s_{2\nu+1} )_{q_1 = (w\tilde{w})^{-1/2} s_{1}}\, ,
\end{align*}
where $E(q_1, q_2)$ is an arbitrary analytic function of $(q_1, q_2)$. 
We see from this equation that a $w$ derivative of a power of $w$ times a function of the self-similar variables $q_1 = (w\tilde{w})^{-1/2} s_1$, $q_2 = (w\tilde{w})^{\nu-1/2} s_{2\nu+1}$ has the same form, with the power of the pre-factor reduced by $1$.
Thus, by induction, the summands of the expression for $h_g(s_1, s_{2\nu+1},\tilde{w})$ are of the form 
\begin{eqnarray*}
&& \tilde{w}^{1-g} \frac{\partial^{j+1}}{\partial s_1\partial w^j} \left[ w^{2- 2g_1} e_{g_1}\left( (w\tilde{w})^{-1/2} s_1, (w\tilde{w})^{\nu-1/2} s_{2\nu+1}\right)\right]_{w=1}\\
&=& \tilde{w}^{1/2-g} \frac{\partial^{j}}{\partial w^j} \left[ w^{3/2- 2g_1} \frac{\partial}{\partial q_1}e_{g_1}\left( (q_1, (w\tilde{w})^{\nu-1/2} s_{2\nu+1}\right)_{q_1 = (w\tilde{w})^{-1/2} s_1}\right]_{w=1}\\
 &=& \tilde{w}^{1/2 - g} E_j( \tilde{w}^{-1/2} s_1, \tilde{w}^{\nu-1/2} s_{2\nu+1}) \,,
\end{eqnarray*}
for a function $E_j$ of the self-similar variables $w^{-1/2} s_1$ and $w^{\nu-1/2} s_{2\nu+1}$. The claims concerning $h$ and $h_g$ now follow from these observations.
\end{proof}

\noindent {\bf Example.} The terms of order less than $1/n^2$ in these series have the coefficients 
\begin{eqnarray} \nonumber 
h_0(s_1, s_{2\nu+1},\tilde{w}) & = & - \tilde{w}\frac{\partial^2}{\partial s_1 \partial w } w^2 e_0( (w\tilde{w})^{-1/2} s_1, (w\tilde{w})^{\nu - 1/2} s_{2\nu+1}) \bigg|_{w=1} \\ \nonumber
& = & - \frac{\partial^2}{\partial s_1 \partial (w\tilde{w}) } (w\tilde{w})^2 e_0( (w\tilde{w})^{-1/2} s_1, (w\tilde{w})^{\nu - 1/2} s_{2\nu+1}) \bigg|_{w\tilde{w} = \tilde{w}} \\ \label{h0term}
& = & - \frac{\partial^2}{\partial s_1 \partial w } w^2 e_0( w^{-1/2} s_1, w^{\nu - 1/2} s_{2\nu+1}) \bigg|_{w = \tilde{w}},
\end{eqnarray}
at order 1 in $h$, where in the second line we have made a change of variables while in the third line we have relabeled $w\tilde{w}$ as $w$; and, 
\begin{eqnarray}
f_0(s_1, s_{2\nu+1},\tilde{w}) & = & \tilde{w}^{2}  \frac{\partial^2}{\partial s_1^2}e_0( \tilde{w}^{-1/2} s_1, \tilde{w}^{\nu - 1/2} s_{2\nu+1} )\nonumber \\
h_1(s_1, s_{2\nu+1}, \tilde{w}) &=& - \frac{1}{2} \frac{\partial^3 }{\partial s_1 \partial w^2} w^2 e_0( (w\tilde{w})^{-1/2} s_1, (w\tilde{w})^{\nu-1/2} ) \bigg|_{w=1}\nonumber \\
&=& - \frac{1}{2} \frac{\partial^3 }{\partial s_1 \partial (w\tilde{w})^2} (w\tilde{w})^2 e_0( (w\tilde{w})^{-1/2} s_1, (w\tilde{w})^{\nu-1/2} ) \bigg|_{w\tilde{w} = \tilde{w}} \nonumber \\
&=& -\frac{1}{2} \frac{\partial^3 }{\partial s_1 \partial {w}^2} {w}^2 e_0( w^{-1/2} s_1, w^{\nu-1/2})\bigg|_{w = \tilde{w}}, \label{h1term}
\end{eqnarray}
at order 1 in $f$ and order $1/n$ in $h$ respectively.

In particular, the terms up through order $1/n$ of $a_{n, N}$ and $b_{n, N}$ are, respectively,
\begin{align} \label{anN-expansion}
x^{1/2} h_0(s_1, s_{2\nu+1}, 1) &= \nonumber
x^{1/2} u_0(s_1,  s_{2\nu+1}) \\
&=  -  x^{1/2} \frac{\partial^2}{\partial s_1 \partial w} w^2 e_0(w^{-1/2} s_1, w^{\nu - 1/2} s_{2\nu+1})|_{w = 1}\\
\nonumber   x^{1/2} h_1(s_1, s_{2\nu+1}, 1) &= 
x^{1/2} u_1(s_1,  s_{2\nu+1}) 
\\&=  - \frac12  x^{1/2} \frac{\partial^3}{\partial s_1 \partial w^2} w^2 e_0(w^{-1/2} s_1, w^{\nu - 1/2} s_{2\nu+1})|_{w = 1}\\
x f_0(s_1, s_{2\nu+1}, 1) &= x  z_0(s_1,  s_{2\nu+1}) \nonumber \\
&= x \frac{\partial^2}{\partial s_1^2}e_0( s_1,  s_{2\nu+1} ).
\end{align}
In Subsection \ref{locsde} we will show that the coefficients $u_0$ and $z_0$ defined here are indeed the same as the functions introduced in Section \ref{sec:1} to describe how the endpoints of the support of the equilibrium measure depend on the parameters in the exponential weight.
\bigskip

We also introduce a shorthand notation to denote the expansion of the coefficients of $h(s_1, s_{2\nu+1},\tilde{w})$ and $f(s_1, s_{2\nu+1},\tilde{w})$ around $w = 1$. This is analogous to what was done in interpreting (\ref{I.004}) via (\ref{Delest}), the main difference being that the order of summation is interchanged. This is again justified by the asymptotic interpretation (\ref{Delest}) where both summations are finite.
\begin{defn} For $\tilde{w} = 1 + k/n$ with  $|k| \leq 2\nu$ and $\frac{2\nu}{n} < \epsilon$,
\begin{eqnarray}
\label{h1k} h(s_1, s_{2\nu+1},\tilde{w}) &=& \sum_{m=0}^\infty \frac{h_{w^{(m)}}|_{w=1}}{m!} \left(\frac{k}{n}\right)^m\\
\label{f1k} f(s_1, s_{2\nu+1},\tilde{w}) &=& \sum_{m=0}^\infty \frac{f_{w^{(m)}}|_{w=1}}{m!} \left(\frac{k}{n}\right)^m
\end{eqnarray}
where the subscript $w^{(m)}$ denotes the formal operation of taking the $m^{th}$ derivative with respect to $w$ of each coefficient of $h$ (respectively $f$):
\begin{eqnarray*}
h_{w^{(m)}} &=& \sum_{g \geq 0} \frac{\partial^{m}}{\partial w^m}  h_g(s_1, s_{2\nu+1},w) \frac{1}{n^g}\\
f_{w^{(m)}}   &=& \sum_{g \geq 0} \frac{\partial^{m}}{\partial w^m}  f_g(s_1, s_{2\nu+1},w) \frac{1}{n^{2g}}
\end{eqnarray*}
As valid asymptotic expansions these representations denote the asymptotic series whose successive terms are gotten by collecting all terms with a common power of $1/n$ in (\ref{h1k}) (respectively (\ref{f1k})). 
\end{defn}

In what follows, in the rest of section 3 and in section 4, we will frequently abuse notation and drop the evaluation at $w=1$. In particular, with $x=1$, we will write 
\begin{eqnarray}
\label{h1kform} a_{n+k, N} &=& \sum_{m=0}^\infty \frac{h_{w^{(m)}}}{m!} \left(\frac{k}{n}\right)^m = \sum_{m=0}^\infty \frac{1}{m!} \sum_{g \geq 0} \frac{\partial^{m}}{\partial w^m}  h_g(s_1, s_{2\nu+1},w) \frac{1}{n^g} \left(\frac{k}{n}\right)^m\\
\label{f1kform} b^2_{n+k, N} &=& \sum_{m=0}^\infty \frac{f_{w^{(m)}}}{m!} \left(\frac{k}{n}\right)^m = \sum_{m=0}^\infty \frac{1}{m!} \sum_{g \geq 0} \frac{\partial^{m}}{\partial w^m}  f_g(s_1, s_{2\nu+1},w) \frac{1}{n^{2g}}\left(\frac{k}{n}\right)^m
\end{eqnarray}
In doing this these series must now be regarded as formal but whose orders are still defined by collecting all terms with a common power of $1/n$.  They will be substituted into the difference string and the Toda equations to derive the respective continuum equations. At any point in this process, if one evaluates these expressions at $w=1$ one may recover valid asymptotic expansions in which the $a_{n+k, N}$ and $b^2_{n+k, N}$ have their original significance as valid asymptotic expansions of the recursion coefficients. In particular, in Section 5, the results of the formal derivations will be evaluated at $w=1$ and we will recover explicit expressions for the $e_g$ appearing in the asymptotic expansion of the partition function. 

\subsection{Continuum Limits of the Difference String Equations} \label{ssec:dse}
We are now in a position to substitute our asymptotic expansions for $a_{n+k,N}$ and $b_{n+k,N}$ into the difference string equations, (\ref{string-star1}) and (\ref{string-star2}). Collecting terms in these equations order by order in powers of $1/n$ we will have a hierarchy of equations that, in principle, allows one to recursively determine the coefficients of (\ref{h1}) and  (\ref{f1}). We will refer to this hierarchy as the {\it Continuum Difference String Equations}. (Note that one has such a hierarchy for each value of $\nu$.) Of course this is a standard procedure in perturbation theory. The equations we will derive are odes in which $w$, now regarded as a continuous variable, is the independent variable. The variables $s_1$ and $s_{2\nu+1}$ here are parameters on which the ode depends analytically. One must still determine, at each level of the hierarchy, which solution of the ode is the one that corresponds to the expressions given for $h_g$ and $f_g$ in Theorem \ref{thm:hf}. This amounts to a kind of solvability condition which will be imposed through a small number of initial Taylor coefficients of $e_g$ to insure that the solution coincides with its enumerative interpretation in terms of counting maps. This solvability analysis will be illustrated in detail in Section \ref{sec:55}. In this subsection the main emphasis will be to derive the form of the continuum string difference equations and their general solutions.
\medskip

From this point on in this section (and in fact for the remainder of the paper) we will set $x = 1$; i.e., $n=N$.  This has the effect of centering the matrix size $n$ at the same scale as that of the potential, $N$. We will also set $s_1 = 0$ from now on since its role in determining the structure of the asymptotic expansions of $a_{n+k}$ and $b_{n+k}$ is now completed.  When $x=1$, $s_{2\nu+1} = t_{2\nu+1}$; however, we will continue to present statements in terms of the $s$-variables. If one wants to subsequently ``detune" to a value $x \lesssim 1$ one can do this by replacing $s_{2\nu+1}$ with its expressions in (\ref{nuscaling2}) and comparing to (\ref{h1}) and (\ref{f1}).  Finally, when the context is clear, we will for simplicity just use $s$ to denote $s_{2\nu+1}$.
\medskip

We begin by substituting the expansions (\ref{h1kform}) and  (\ref{f1kform}) into the difference equations, (\ref{string-star1}) and (\ref{string-star2}), satisfied by these coefficients (as represented through $\mathcal{L}$). We arrive at the following formal asymptotic equations.
For equation (\ref{string-star1}) one has:
\begin{align} \label{string-expansion-1}
0 &=  \left[ \sum_{m=1}^\infty \frac{h_{w^{(m)}}}{m!} \left( \frac{1}{n} \right)^m \right] \left\{\left[ \sum_{m=0}^\infty \frac{f_{w^{(m)}}}{m!} \left( \frac{1}{n} \right)^m \right] \right.
\nonumber \\&\phantom{=}
 + (2\nu+1) s_{2\nu+1} \left[ \sum_{P \in \mathcal{P}^{2\nu}(1, 0) }
\right. \\ &\phantom{=} \left.\left.
\left[ \prod_{p_a = 1}^{2\mu(P)+1} \sum_{m=0}^\infty \frac{h_{w^{(m)}}}{m!} \left( \frac{\ell_{p_a} }{n} \right)^m \right] 
\left[ \prod_{p_b = 1}^{\nu - \mu(P)} \sum_{m=0}^\infty \frac{f_{w^{(m)}}}{m!} \left( \frac{\ell_{p_b} }{n} \right)^m \right] \right] \right\}
\nonumber \\&\phantom{=}
+ (2\nu+1) s_{2\nu+1} \left[ \sum_{P \in \mathcal{P}^{2\nu}(2,0) }
\nonumber \right. \\&\phantom{=}
\left[ \prod_{p_a = 1}^{2\mu(P)} \sum_{m=0}^\infty \frac{h_{w^{(m)}}}{m!} \left( \frac{\ell_{p_a}}{n} \right)^m \right] 
\left[ \prod_{p_b = 1}^{\nu - \mu(P) +1} \sum_{m=0}^\infty \frac{f_{w^{(m)}}}{m!} \left( \frac{\ell_{p_b} }{n} \right)^m \right] 
\nonumber  \\ &\phantom{=} \left.
- \left[ \prod_{p_a = 1}^{2\mu(P)} \sum_{m=0}^\infty \frac{h_{w^{(m)}}}{m!} \left( \frac{\ell_{p_a} -1 }{n} \right)^m \right]
 \left[ \prod_{p_b = 1}^{\nu - \mu(P) + 1} \sum_{m=0}^\infty \frac{f_{w^{(m)}}}{m!} \left( \frac{\ell_{p_b} -1 }{n} \right)^m \right] \right]
\end{align}
where $\mu(P) = \lfloor \sigma/2 \rfloor$ for $\sigma$ equal to the total number of horizontal steps in a given path $P$ and 
$\ell_{p_a}$ (respectively $\ell_{p_b}$) denotes the lattice location of the path at the $p_a^{th}$ horizontal step (respectively before the $p_b^{th}$ downstep). Note also that we have taken advantage of the {\it discrete space homogeneity} of these walks in order to shift the initial/final 
points of these paths to $(1,0), (2,0)$ and $(2,1)$ in the respective cases. Note further that each term of \eqref{string-expansion-1} is divisible by 
\begin{equation*}
b_{n+1}^2 = \sum_{m=0}^\infty \frac{f_{w^{(m)}}}{m!} \left( \frac{1}{n} \right)^m \,.
\end{equation*}

Likewise for (\ref{string-star2}) we find
\begin{align} \label{string-expansion-2}
\frac{1}{n} &= \left[ \sum_{m=1}^\infty \frac{f_{w^{(m)}}}{m!}  \left( \frac{ 1}{n} \right)^m \right] 
\nonumber \\ &\phantom{=} 
+ (2\nu+1) s_{2\nu+1} \left[ \sum_{P \in \mathcal{P}^{2\nu}(1,0) } \right. \nonumber
\\ &\phantom{=}
 \left[ \prod_{p_a = 1}^{2\mu(P)+1} \sum_{m=0}^\infty \frac{h_{w^{(m)}}}{m!} \left( \frac{\ell_{p_a}}{n} \right)^m \right] \left[ \prod_{p_b = 1}^{\nu - \mu(P)} \sum_{m=0}^\infty \frac{f_{w^{(m)}}}{m!} \left( \frac{\ell_{p_b} }{n} \right)^m \right] 
\nonumber  \\ &\phantom{=} \left. 
- \left[ \prod_{p_a=1}^{2\mu(P)+1} \sum_{m=0}^\infty \frac{h_{w^{(m)}}}{m!} \left( \frac{\ell_{p_a} -1 }{n} \right)^m \right] 
\left[ \prod_{p_b=1}^{\nu - \mu(P)} \sum_{m=0}^\infty \frac{f_{w^{(m)}}}{m!} \left( \frac{\ell_{p_b} -1}{n} \right)^m \right] \right]\, .
\end{align}
\medskip

To begin with, the equations at leading order are
\begin{eqnarray*} \label{leadstring}
0 &=& \partial_w h_0 \left( 1 + (2\nu + 1)s_{2\nu+1} \sum_{\mu = 0}^{\nu-1} {2\nu \choose 2\mu + 1, \nu-\mu-1, \nu-\mu} h_0^{2\mu+1} f_0^{\nu - \mu-1}\right)\\
&+& (2\nu + 1)s_{2\nu+1} \sum_{\mu = 0}^{\nu-1}{2\nu \choose 2\mu , \nu-\mu-1, \nu-\mu +1} \partial_w \left(h_0^{2\mu} f_0^{\nu - \mu+1}\right) / f_0\,,\\
1 &=& \partial_w f_0 +  (2\nu + 1)s_{2\nu+1} \sum_{\mu = 0}^{\nu-1}{2\nu \choose 2\mu + 1 , \nu-\mu-1, \nu-\mu } \partial_w \left(h_0^{2\mu+1} f_0^{\nu - \mu}\right)\,.
\end{eqnarray*}

This may be written in a vector form as 
\begin{prop}
\begin{equation} \label{sdgnu}
 {\bf A} \begin{pmatrix} h_{0, w} \\ f_{0, w} \end{pmatrix}  = \begin{pmatrix} 0 \\ 1 \end{pmatrix}
\end{equation}
 where
\begin{align}
{\bf A}_{11} &= 1 + (2\nu+1) s_{2\nu+1} \sum_{\mu=0}^{\nu-1} \binom{2\nu}{2\mu+1, \nu-\mu-1, \nu-\mu} h_0^{2\mu+1} f_0^{\nu-\mu-1} \\
&\phantom{= f_0 } \nonumber
+ (2\nu+1) s_{2\nu+1} \sum_{\mu=0}^{\nu-1} \binom{2\nu}{2\mu, \nu-\mu-1, \nu-\mu+1} (2\mu) h_0^{2\mu-1} f_0^{\nu-\mu}\,,  \\
{\bf A}_{12} &= (2\nu+1) s_{2\nu+1} \sum_{\mu=0}^{\nu-1} \binom{2\nu}{2\mu, \nu-\mu-1, \nu-\mu+1} (\nu-\mu+1) h_0^{2\mu} f_0^{\nu-\mu-1} \,, \\
{\bf A}_{21} &= f_0 {\bf A}_{12} \,, \label{a21}\\
{\bf A}_{22} &= 1 + (2\nu+1) s_{2\nu+1} \sum_{\mu=0}^{\nu-1} \binom{2\nu}{2\mu+1, \nu-\mu-1, \nu-\mu} (\nu-\mu) h_{0}^{2\mu+1} f_0^{\nu-\mu-1} = {\bf A}_{11} \,,
\end{align}
in terms of tri-nomial coefficients.
\end{prop}
\bigskip

We will see that the coefficient matrix of the $n^{-2g-1}$ terms is of the same form as the matrix in equation (\ref{sdgnu}).  Thus the following lemma will be useful:
\begin{lem} \label{lemma-inverse}
For the matrix ${\bf A}$ given in \eqref{sdgnu} 
\begin{equation*}
{\bf A}^{-1} = 
\begin{pmatrix} f_{0, w} & h_{0, w} \\ h_{0, w} f_0 & f_{0, w} \end{pmatrix}\,.
\end{equation*}
\end{lem}
\begin{proof}
The second column of the inverse follows directly from (\ref{sdgnu}).  
To find the first column one notes that 
\begin{equation*}
{\bf A}^{-1} = \begin{pmatrix} {\bf A}_{11} & {\bf A}_{12} \\ {\bf A}_{12} f_0 & {\bf A}_{11} \end{pmatrix}^{-1} = 
\frac{1}{ {\bf A}_{11}^2 - {\bf A}_{12}^2 f_0} \begin{pmatrix} {\bf A}_{11} & -{\bf A}_{12} \\ -{\bf A}_{12} f_0 & {\bf A}_{11} \end{pmatrix} \,.
\end{equation*}
Thus we have that ${\bf A}_{11}/\mbox{det}({\bf A}) = f_{0, w}$ and $-{\bf A}_{12}/\mbox{det}({\bf A}) = h_{0, w}$, and the result follows.
\end{proof}

We will also need the following lemma:
\begin{lem} \label{lemma37}
For the matrix ${\bf A}$ given in \eqref{sdgnu}
\begin{align*}
\delta_{h_0} {\bf A}_{12} - \delta_{f_0} {\bf A}_{11} &= 0 \\
\delta_{h_0} {\bf A}_{22} - \delta_{f_0} {\bf A}_{21} &= 0 \,,
\end{align*}
where $\delta_{h_0}$ (resp. $\delta_{f_0}$) denotes the functional derrivative with respect to $h_0$ (resp. $f_0$).  
\end{lem}
\begin{proof}
The left-hand side of the first equation, written out, is 
\begin{align} \nonumber
(2\nu +1) s_{2\nu+1} & \left[ \sum_{\mu=1}^{\nu-1} \binom{2\nu}{2\mu, \nu-\mu-1, \nu-\mu+1} (\nu-\mu+1) (2\mu) h_0^{2\mu-1} f_0^{\nu-\mu-1} 
\right. \\ \label{comp-1A} & \qquad
- \sum_{\mu=0}^{\nu-1} \binom{2\nu}{2\mu+1, \nu-\mu-1, \nu-\mu} (\nu-\mu-1) h_0^{2\mu+1} f_0^{\nu-\mu-2} 
\\  \nonumber & \qquad \left.
- \sum_{\mu=1}^{\nu-1} \binom{2\nu}{2\mu, \nu-\mu-1, \nu-\mu+1} (2\mu) (\nu - \mu) h_0^{2\mu-1} f_0^{\nu-\mu-1} \right] \,.
\end{align}
Shifting the index of the middle sum by $\mu \mapsto \mu-1$ one sees that the coefficient of each monomial in $h_0, f_0$ cancels and the result follows.
To prove the second formula one first applies the identity \eqref{a21} to find 
\begin{equation} \label{3.28}
\delta_{h_0} {\bf A}_{22} - \delta_{f_0} {\bf A}_{21} = \delta_{h_0} {\bf A}_{22} - f_0 \delta_{f_0} {\bf A}_{12} - {\bf A}_{12} \,.
\end{equation}
Written out, the right-hand side of \eqref{3.28} is
\begin{align*}
(2\nu+1) s_{2\nu+1} \sum_{\mu=1}^{\nu-1} & \left[ 
\binom{2\nu}{2\mu+1, \nu-\mu-1, \nu-\mu} (\nu-\mu) (2\mu+1) \right. \\  
&- \binom{2\nu}{2\mu, \nu-\mu-1, \nu-\mu+1} (\nu-\mu+1) (\nu-\mu-1)  \\
&\left. - \binom{2\nu}{2\mu, \nu-\mu-1, \nu-\mu+1} (\nu-\mu+1) \right] h_0^{2\mu} f_0^{\nu-\mu-1} \, ,
\end{align*}
whose coefficients manifestly vanish.
\end{proof}

The homogenous terms of the equations at level $n^{-2g-1}$ can be computed directly.  They are linear in $h_{2g}$ and $f_g$ with coefficients 
depending only on $h_0$, $f_0$ and their $w$ derivatives. The inhomogeneous (forcing) terms depend on $h_j$ for $j<2g$, and $f_j$ for $j<g$.  As usual in perturbation theory, the homogeneous part of the equation can be derived by replacing $(h_0, f_0)$ in the leading order equations with 
\begin{equation} 
( h_0 + \epsilon h_{2g}, f_0 + \epsilon f_g )\,,
\end{equation} 
and retaining just the first order in $\epsilon$ terms.  We find that the homogeneous terms are
\begin{equation} \label{blue-1}
{\bf A} \begin{pmatrix} h_{2g, w} \\ f_{g, w} \end{pmatrix} + 
h_{2g} \delta_{h_0} {\bf A} \begin{pmatrix} h_{0, w} \\ f_{0, w} \end{pmatrix} + 
f_g \delta_{f_0} {\bf A} \begin{pmatrix} h_{0, w} \\ f_{0, w} \end{pmatrix}  \,.
\end{equation} \pagebreak[3]

\noindent We also have the identity 
\begin{equation} \label{blue-2}
\partial_w \left\{ {\bf A} \begin{pmatrix} h_{2g} \\ f_g \end{pmatrix} \right\} = 
{\bf A} \begin{pmatrix} h_{2g, w} \\ f_{g, w} \end{pmatrix} + 
 \delta_{h_0} {\bf A} \begin{pmatrix} h_{0, w} h_{2g} \\ h_{0, w} f_g \end{pmatrix} +
\delta_{f_0} {\bf A} \begin{pmatrix} f_{0, w} h_{2g} \\ f_{0, w} f_g \end{pmatrix} \,.
\end{equation}
Lemma \ref{lemma37} implies that \eqref{blue-1} is in fact equal to the right-hand side of \eqref{blue-2}.  
Thus the equation at order $n^{-2g-1}$ has the form
\begin{equation}\label{n2gm1-A}
\partial_w \left\{ {\bf A} \begin{pmatrix} h_{2g} \\ f_g \end{pmatrix} \right\} = \begin{pmatrix} F_{2g}^{(1)} \\ F_{2g}^{(2)} \end{pmatrix}  \,,
\end{equation}
where $F_{2g}^{(1)}$ and $F_{2g}^{(2)}$ are expressions invovling the lower order terms in the expansions of \eqref{string-expansion-1} and \eqref{string-expansion-2}.
We also find, at order $n^{-2g}$, that
\begin{equation} \label{n2g-B}
\partial_w \left\{ {\bf A}_{11} h_{2g-1} \right\} = F_{2g-1}^{(1)} \,.
\end{equation}
In fact there is a second equation involving $h_{2g-1}$; however, it must be equivalent to the first and so we do not record it.  
Equations \eqref{n2gm1-A} and \eqref{n2g-B}, together with Lemma \ref{lemma-inverse} yield the following 
\begin{prop} \label{lemma-solution}
The functions $h_{2g}, h_{2g-1}, f_g$ may be recursively found by the formulas
\begin{align} \label{IntFact}
\left(
\begin{array}{c}
 h_{2g} \\
  f_{g}
\end{array}
\right) &= 
\begin{pmatrix} f_{0, w} & h_{0, w} \\ h_{0, w} f_0 & f_{0, w} \end{pmatrix} 
\int \left(
\begin{array}{c}
F^{(1)}_{2g} \\
F^{(2)} _{2g} 
\end{array}
\right) dw \,,\\
\label{IntFact-odd}  h_{2g-1} &= \frac{1}{{\bf A}_{11} } \int  F^{(1)}_{2g-1} dw.
\end{align}
\end{prop}

\subsection{The Continuum Limit of the Toda equations}
Analogous to what was done in the previous subsection for the difference string equations, one can study the system
(\ref{an}) and (\ref{bn}) expanded on the formal asymptotic series  (\ref{h1kform}) and (\ref{f1kform}):
\begin{align}
\nonumber - \frac{1}{n} \frac{d}{ds}h(s,w) &= \sum_{P \in \{\mathcal{P}^{2\nu + 1}(1,0)\}} 
\\ &\phantom{=} \nonumber
\left[ \prod_{p_a=1}^{2\mu(P)}
\sum_{m=0}^\infty \frac{h_{w^{(m)}}}{m! }\left(\frac{\ell_{p_a} }{n}\right)^m \right] \left[\prod_{p_b=1}^{\nu - \mu(P)+1}
 \sum_{m=0}^\infty \frac{f_{w^{(m)}}}{m! }\left(\frac{\ell_{p_b}}{n}\right)^m \right] \\
\label{h} &\phantom{=} - \left[ \prod_{p_a=1}^{2\mu(P)}
\sum_{m=0}^\infty \frac{h_{w^{(m)}}}{m! }\left(\frac{\ell_{p_a} -1}{n}\right)^m \right] \left[\prod_{p_b=1}^{\nu - \mu(P)+1}
 \sum_{m=0}^\infty \frac{f_{w^{(m)}}}{m! }\left(\frac{\ell_{p_b} -1 }{n}\right)^m \right] 
\end{align}
\begin{align}
\nonumber - \frac{1}{n} \frac{d}{ds}f(s,w)&=
\sum_{P \in \{\mathcal{P}^{2\nu + 1}(2,0)\}} 
\\ &\phantom{=} \nonumber
\left[ \prod_{p_a=1}^{2\mu(P)+1}
\sum_{m=0}^\infty \frac{h_{w^{(m)}}}{m! }\left(\frac{\ell_{p_a} -1 }{n}\right)^m \right] \left[\prod_{p_b=1}^{\nu - \mu(P) +1}
 \sum_{m=0}^\infty \frac{f_{w^{(m)}}}{m! }\left(\frac{\ell_{p_b} -1 }{n}\right)^m \right] \\
\label{f} &\phantom{=} - \left[ \prod_{p_a=1}^{2\mu(P)+1}
\sum_{m=0}^\infty \frac{h_{w^{(m)}}}{m! }\left(\frac{\ell_{p_a} -2}{n}\right)^m \right] \left[\prod_{p_b=1}^{\nu - \mu(P)+1}
 \sum_{m=0}^\infty \frac{f_{w^{(m)}}}{m! }\left(\frac{\ell_{p_b} - 2 }{n}\right)^m \right] \\
\nonumber &\phantom{=} 
 -  \left( \sum_{m=1}^\infty \frac{h_{w^{(m)}}}{m! } \left(\frac{-1}{n}\right)^m  \right) 
 \sum_{P \in \{\mathcal{P}^{2\nu+1}(1,0)\}} \left[ \prod_{p_a=1}^{2\mu(P)}
\sum_{m=0}^\infty \frac{h_{w^{(m)}}}{m! }\left(\frac{\ell_{p_a} -1}{n}\right)^m\right] \\
\nonumber &\phantom{=\sum_{P \in \{\mathcal{P}^{2\nu + 1}(2,0)\}}} \hspace{0.5in}
\times \left[\prod_{p_b=1}^{\nu - \mu(P)+1}
 \sum_{m=0}^\infty \frac{f_{w^{(m)}}}{m! }\left(\frac{\ell_{p_b} -1}{n}\right)^m \right] 
 \end{align}
where, again, $\mu(P) = \lfloor \sigma/2 \rfloor$ for $\sigma$ equal to the total number of horizontal steps in a given path $P$  and 
$\ell_{p_a}$ (respectively $\ell_{p_b}$) denotes the lattice location of the path at the $p_a^{th}$ horizontal step (respectively before
 the $p_b^{th}$ downstep). 
 As before, the expansion order by order produces a hierarchy of equations that we will call the {\it Continuum Toda equations}.  At leading order in the hierarchy one has, for general $\nu$,
\begin{eqnarray} \label{eqh0}
- \frac{d}{ds}h_0(s,w) &=& \partial_w \sum^\nu_{\mu = 0} {2\nu + 1 \choose 2\mu, \nu - \mu, \nu - \mu + 1} h_0^{2\mu} f_0^{\nu - \mu + 1}\\ \label{eqf0}
 - \frac{d}{ds}f_0(s,w) &=& \partial_w \sum^{\nu-1}_{\mu = 0} {2\nu + 1 \choose 2\mu +1, \nu - \mu - 1, \nu - \mu + 1} h_0^{2\mu+1} f_0^{\nu - \mu +1} +
\\ && \nonumber
+  \partial_w h_0 \sum^\nu_{\mu = 0} {2\nu + 1 \choose 2\mu , \nu - \mu, \nu - \mu+1} h_0^{2\mu} f_0^{\nu - \mu + 1}\, .
\end{eqnarray}
\medskip

At orders $n^{-2g}$ the equations are equivalent to a hierarchy of $2 \times 2$ quasi-linear 
systems of PDE:
\begin{prop}
\begin{equation} \label{conslaw}
-\frac{d}{ds} \begin{pmatrix} h_{2g} \\ f_g \end{pmatrix} =
\partial_w \left[ {\bf B} \begin{pmatrix} h_{2g} \\ f_g \end{pmatrix} \right]
+ \begin{pmatrix} 0 \\ - r_1 f_{0, w} h_{2g} + r_2 h_{0, w} f_g \end{pmatrix}
+ \left(
\begin{array}{c}
  \mbox{Forcing}_g^{(1)}  \\
  \mbox{Forcing}_g^{(2)}
\end{array}
\right)
\,,
\end{equation}
where
\begin{align}
{\bf B}_{11} = {\bf B}_{22} &= \sum_{\mu=1}^{\nu} \binom{2\nu+1}{2\mu, \nu-\mu, \nu-\mu+1} 2\mu h_{0}^{2\mu-1} f_0^{\nu-\mu+1}\,, \\
{\bf B}_{12} &= \sum_{\mu=0}^{\nu} \binom{2\nu+1}{2\mu, \nu-\mu, \nu-\mu+1} (\nu-\mu+1) h_{0}^{2\mu} f_0^{\nu-\mu} \,, 
\end{align}
\begin{align}
{\bf B}_{21} &= f_0 {\bf B}_{12} \,,\\
r_1 &=(2\nu+1) h_0^{2\nu} +  \sum_{\mu=0}^{\nu-1} \binom{2\nu+1}{2\mu+1, \nu-\mu-1, \nu-\mu+1} (2\mu+1) (\nu-\mu+1) h_0^{2\mu} f_0^{\nu-\mu} \,, \\
r_2 &= \sum_{\mu=0}^{\nu} \binom{2\nu+1}{2\mu, \nu-\mu, \nu-\mu+1} (\nu-\mu+1) h_0^{2\mu} f_0^{\nu-\mu} \, .
\end{align} 
\end{prop}
\begin{proof}
The proposition follows from the same approach and methods used to establish Proposition \ref{lemma-solution}.  The various relationships between the entries of ${\bf B}$ are consequences of simple trinomial identities as was the case for the relations between entries of ${\bf A}$.
\end{proof}

\begin{rem}\label{rem:3.5} We observe that the homogeneous terms in (\ref{conslaw}) are almost in pure conservation law form, with the first terms on the right being in the form of a {\it spatial} derivative of a flux pair associated to the density pair $(h_{2g}, f_g)$. The form of the non-conservative homogeneous terms (the second group of terms) suggests that one might be able to use the difference string equations to rewrite these in terms of lower genus expressions and thereby pass them into the forcing so that the equations would then have the structure of a hierarchy of forced conservation laws. Indeed, in Section \ref{sec:55} we show that this is what happens for the genus 1 equations.

We also note that the numerical coefficients appearing in the homogeneous terms of (\ref{conslaw})  depend only on the total number of Motzkin paths of each class $\mathcal{P}^j(m_1, m_2)$ that appear in the Toda equations. On the other hand the specific terms in the forcing expressions depend on more detailed combinatorial characteristics of these Motzkin paths.  See \cite{EMP08, Er09} for details about how the forcing terms can be determined explicitly from the structure of the lattice paths in the case of even valence. 
\end{rem}

\section{Specialization to the Trivalent Case} \label{sec:55}

We now illustrate the results of section \ref{sec:44}, in the trivalent case (when $\nu=1$) to demonstrate their form and utility.  We will then pick up from lemma \ref{lemma-solution} in this case and find explicit expressions for $h_1, h_2$ and $f_{1}$ in terms of $f_0$, $h_0$ and their $w$-derivatives.  As mentioned at the start of subsection \ref{ssec:dse}, this will require comparison with the enumerative significance of the coefficients in the asymptotic expansions in order to determine a unique solution. Note in particular that we will need to go beyond the results of section \ref{sec:44} in order to explicitly determine the forcing terms $F_1^{(1)}$, $F_2^{(1)}$ and $F_2^{(2)}$. In this section we will use $s$ to denote $s_3$.

\subsection{String difference Equations for $j  = 3$} 
In the {\it trivalent case} one has
\begin{align} \label{eq-string-1}
0 &= b_{n+1}^2 \left[ (a_{n+1} - a_n) ( 1 + 3 s (a_n + a_{n+1})) + 3 s ( b_{n+2}^2 - b_n^2) \right] \,,\\ \label{eq-string-2}
\frac{1}{n} &= ( b_{n+1}^2 + 3 s b_{n+1}^2 ( a_n + a_{n+1}) ) - ( b_n^2 + 3 s b_n^2 (a_n + a_{n-1})) 
\\
&= ( b_{n+1}^2 - b_n^2) + 3 s b_{n+1}^2 (a_n + a_{n+1}) - 3 s b_n^2 (a_n + a_{n-1}) 
\,.
\end{align}

Using (\ref{a}) and (\ref{b}), and Taylor expansions around $w=1$ of the continuum limits we have:
\begin{itemize}

\item Dividing (\ref{eq-string-1}) by $b_{n+1}^2$:
\begin{align} \nonumber
0 &= ( h(s, 1+\frac{1}{n}) - h(s, 1) ) ( 1 + 3 s ( h(s, 1) + h(s, 1+\frac{1}{n}) ) 
+ 3 s ( f(s, 1+\frac{2}{n}) - f(s, 1) ) 
\\
&= \left( \sum_{m=1}^\infty \frac{h_{w^{(m)}}}{m!} \frac{1}{n^m} \right) \left[ 
1 + 3 s \left( 2 h + \sum_{m=1}^\infty \frac{h_{w^{(m)}}}{m!} \frac{1}{n^m} \right) \right] 
+ 3 s \left( \sum_{m=1}^\infty \frac{f_{w^{(m)}}}{m!} \frac{2^m}{n^m} \right) \,,
\label{eq-string-3}
\end{align}
where the second line is evaluated at $w=1$.

\item The equation (\ref{eq-string-2}) becomes:
\begin{align} \nonumber
\frac{1}{n} &=  ( f(s, 1+\frac{1}{n}) - f(s, 1) ) + 3 s f(s, 1+\frac{1}{n}) ( h(s, 1) + h(s, 1+\frac{1}{n}) )  
\\ &\phantom{=} \hspace{1cm} \nonumber
- 3 s f(s, 1) ( h(s, 1) + h(s, 1-\frac{1}{n})) 
\\
&= \left( \sum_{m=1}^\infty \frac{f_{w^{(m)}}}{m!} \frac{1}{n^m} \right) + 
3 s \left( f + \sum_{m=1}^\infty \frac{f_{w^{(m)}}}{m!} \frac{1}{n^m} \right) \left( 2 h + \sum_{m=1}^\infty \frac{h_{w^{(m)}}}{m!} \frac{1}{n^m} \right) 
\nonumber \\ &\phantom{=} \hspace{1cm} 
- 3 s f \left( 2 h + \sum_{m=1}^\infty \frac{h_{w^{(m)}}}{m!} \frac{(-1)^m}{m!} \right) \,,
\label{eq-string-4}
\end{align}
where the second line is evaluated at $w=1$.
\end{itemize}

\subsubsection{Leading Order \label{locsde}}

The leading order ($\mathcal{O}(n^{-1})$) of the system (\ref{eq-string-3}-\ref{eq-string-4}) is as given in \eqref{sdgnu} with $\nu=1$:
\begin{equation} \label{leading-order-system}
\begin{pmatrix} 0 \\ 1 \end{pmatrix} = 
\begin{pmatrix} 1 + 6 s h_0 & 6 s \\ 6s f_0 & 1 + 6s h_0 \end{pmatrix}
\begin{pmatrix} h_{0, w} \\ f_{0, w} \end{pmatrix}\,.
\end{equation} 
 Expanding we have 
\begin{align*}
0 &= (1+ 6 s h_0 ) h_{0, w} + 6 s f_{0, w} \\
1 &= 6 s f_0 h_{0, w} + (1+ 6s h_0) f_{0, w} \,.
\end{align*}
These can be anti-differentiated with respect to $w$:
\begin{align*}
C_1(s) &=  h_0 + 3 s h_0^2 + 6 s f_0  \\
w + C_2(s) &=  6 s f_0 h_0 + f_0\,. 
\end{align*}
The $C_1(s)$ and $C_2(s)$ are constants of integration which must be determined by the combinatorial interpretation or some other constraints.  For example converting to the self-similar variables, $\tilde{s} = w^{1/2} s$, dividing the first equation by $w^{1/2}$ and the second by $w$, we find
\begin{align*}
w^{-1/2} C_1(s) &= u_0(\tilde{s} ) + 3 \tilde{s} u_0(\tilde{s})^2 + 6 \tilde{s} z_0(\tilde{s}) \\
1 + w^{-1} C_2(s) &= 6 \tilde{s} z_0( \tilde{s}) u_0(\tilde{s}) + z_0(\tilde{s}) \,.
\end{align*}
In order for the left hand side of these equations to give functions of the self-similar variable $\tilde{s}$ we must have that $C_1(s) = c_1 s^{-1}$ and $C_2(s) = c_2 s^{-2}$.  However the functions on the right hand side are analytic in a neighborhood of $s = s_3 = 0 $ and so we conclude that $c_1=c_2 = 0$.
Comparing these with \eqref{ideal1} and \eqref{ideal2}, we have shown here that the leading order functions of the asymptotic expansion of $a_{n, N}$ and $b_{n, N}$ agree with the functions $u_0$ and $z_0$ describing the equilibrium measure.

\subsubsection{$n^{-2g}$ terms} \label{n(-2g)}

The odd terms of the expansion for $h(s, w)$ are governed by either of the $n^{-2g}$ terms of equations (\ref{eq-string-3}-\ref{eq-string-4}); i.e.,   by either of the equations
\begin{align*}
0 &= h_{2g-1, w} \left[ 1 + 6 s h_0 \right] + h_{0, w} \left[  6 s h_{2g-1} \right] - F^{(1)}_{2g-1} \\
0 &= 6 s f_0 h_{2g-1, w} + 6 s f_{0, w} h_{2g-1} - F^{(2)}_{2g-1} \,,
\end{align*}
where the $F^{(j)}_{2g-1}$ are the forcing terms coming from the $n^{-2g}$ terms of (\ref{eq-string-3}-\ref{eq-string-4}) which do not contain an $h_{2g-1}$ or its derrivatives.  The first equation is equivalent to the specialization of \eqref{n2g-B} to the trivalent case:
\begin{eqnarray*}
\partial_w \left\{\left[ 1 + 6s h_0 \right] h_{2g-1} \right\} &=& F^{(1)}_{2g-1}\,.
\end{eqnarray*}

\subsubsection{$n^{-2g-1}$ terms}

The even terms of the expansion for $h(s, w)$ and the terms of the expansion for $f(s, w)$ are governed by the $n^{-2g-1}$ terms of equations (\ref{eq-string-3}-\ref{eq-string-4}).  We find the system (given by taking $\nu=1$ in \eqref{n2gm1-A})
\begin{equation}
\partial_w \left\{ \begin{pmatrix} 1 + 6 s h_0 & 6 s \\ 6 s f_0 & 1 + 6 s h_0 \end{pmatrix} \begin{pmatrix} h_{2g} \\ f_g \end{pmatrix} \right\} = \begin{pmatrix} F^{(1)}_{2g} \\ F^{(2)}_{2g} \end{pmatrix} 
\end{equation}
where $F^{(j)}_{2g}$ are the forcing terms coming from the $n^{-2g-1}$ terms of (\ref{eq-string-3}-\ref{eq-string-4}) which do not contain an $h_{2g}$ or $f_g$ or their derrivatives.

Applying Lemma \ref{lemma-inverse} and Proposition \ref{lemma-solution} we have
\begin{eqnarray} 
\left(
\begin{array}{c}
 h_{2g} \\
  f_{g}
\end{array}
\right) &=& 
\begin{pmatrix} f_{0, w} & h_{0, w} \\ h_{0, w} f_0 & f_{0, w} \end{pmatrix} 
\int \left(
\begin{array}{c}
F^{(1)}_{2g} \\
F^{(2)} _{2g} 
\end{array}
\right) dw.\\
 h_{2g-1} &=& \frac{1}{1 + 6s h_0 } \int  F^{(1)}_{2g-1} dw.
\end{eqnarray}

\begin{lem}
$h_1(s,w) = \frac12 h_{0,w}(s,w)$ 
\end{lem}
\begin{proof}
From (\ref{IntFact-odd}) and the $n^{-2}$ coefficients in the first difference string equation \eqref{eq-string-3} one has
\begin{eqnarray*}
\left[ 1 + 6s h_0 \right] h_{1}  &=& - \int \left(3s h_{0,w}^2 + \frac12 (1 + 6s h_0) h_{0,ww} + 6s f_{0,ww}\right) dw\\
&=& - \left(\frac12 (1 + 6s h_0) h_{0,w} + 6s f_{0,w}\right) + C(s)\\
&=&  \frac12 \left( (1 + 6s h_0) h_{0,w}\right) + C(s) \,\, \mbox{by (\ref{leading-order-system})} \\
h_1 &=&  \frac12 h_{0,w} + \frac{C(s)}{1 + 6s h_0}\,.
\end{eqnarray*}
We see that this agrees with the first two terms of the asymptotic expansion of $a_{n+k, N}$, given in \eqref{h0term} and \eqref{h1term},
\begin{align}
h_0(s,w) = -\frac{\partial^2}{\partial s_1 \partial w} w^2 e_0\left( w^{-1/2} s_1, w^{1/2} s_3\right) \bigg|_{s_1 = 0} \\
h_1(s,w) = - \frac{1}{2} \frac{\partial^3}{\partial s_1 \partial w^2} w^2 e_0\left( w^{-1/2} s_1, w^{1/2} s_3\right) \bigg|_{s_1=0} \,,
\end{align}
from which we can also conclude that $C(s) \equiv 0$. 
\end{proof}

\begin{prop} \label{prop.h2.f1}
\begin{eqnarray*}
\left(
\begin{array}{c}
 h_{2} \\
  f_{1}
\end{array}
\right) &=&
- \begin{pmatrix} f_{0,w} & h_{0, w} \\ h_{0,w} f_0 & f_{0,w} \end{pmatrix}
\left(
\begin{array}{c}
\frac{13}4 s(h_{0,w})^2 + \frac52 s h_0 h_{0,ww} + 4s f_{0,ww} + \frac5{12} h_{0,ww} \\
 0
\end{array}
\right) 
\end{eqnarray*}
\end{prop}
\noindent {\bf Proof.} 
From (\ref{IntFact}) and the $n^{-3}$ coefficients in the differenced string equations (\ref{eq-string-3}-\ref{eq-string-4}) one has
\begin{eqnarray*}
&&\left(
\begin{array}{c}
 h_{2} \\
  f_{1}
\end{array}
\right) = 
-  \begin{pmatrix} f_{0,w} & h_{0, w} \\ h_{0,w} f_0 & f_{0,w} \end{pmatrix} 
\times \\ &&
\int \bigg(
\begin{array}{c}
 6s (h_1 + h_{0,w}) \left(h_{1,w} + \frac12 h_{0,ww}\right) + (1 + 6s h_0)  \left(\frac12 h_{1,ww} + \frac16 h_{0,www}\right) + \\
 \frac16 (1 + 6s h_0) f_{0, www} + \frac32 s f_{0,w} (2h_{1,w} + h_{0,ww}) + \frac32 s f_{0,ww} (2h_{1} + h_{0, w}) + 
\end{array}
  \\
&& \hspace{5cm} 
\begin{array}{c}
+ 4s f_{0, www}\\ + s f_0 h_{0, www} 
\end{array} \bigg) dw.
\end{eqnarray*}
Substituting for the $h_1$ on the right-hand side, from the previous lemma, one finds that the integrand is an exact derivative so that the right-hand side becomes:
 \begin{equation} \label{RHS}
- \begin{pmatrix} f_{0,w} & h_{0, w} \\ h_{0,w} f_0 & f_{0,w} \end{pmatrix}\begin{pmatrix} \frac{13}4 s(h_{0,w})^2 + \frac52 s h_0 h_{0,ww} + 4s f_{0,ww} + \frac5{12} h_{0,ww} + C_1(s) \\
2s h_{0,w} f_{0,w} + s(h_0  f_{0,ww} + h_{0,ww} f_0 ) +  \frac16 f_{0,ww} +
 C_2(s) 
\end{pmatrix}\,.
\end{equation}
One finds further, using the leading order equations (\ref{leading-order-system}), that the second entry of the right vector in (\ref{RHS}) is identically zero, except possibly for the constant of integration $C_2(s)$.

Converting to equations in the self-similar variable $\tilde{s} = w^{1/2} s$, and focusing on the terms involving $C_1$ and $C_2$  in the expression for $f_1$ (the second component in (\ref{RHS})),  we have
\begin{eqnarray*}
f_1(\tilde{s}) &=& w^{-1} z_1(\tilde{s})\\
&=& \left[\mbox{terms not involving}\,\, C_1, C_2\right] - \left\{( f_0 h_{0,w}) C_1(s) + (f_{0,w}) C_2(s)\right\} \\
z_1(\tilde{s}) &=&  \left[\mbox{terms not involving}\,\, C_1, C_2\right] - \frac12 z_0(\tilde{s}) (u_0(\tilde{s}) + \tilde{s} u_0'(\tilde{s}) ) w^{3/2} C_1(s) 
\\ && \hspace{4cm}
- (z_0(\tilde{s}) + \frac12 \tilde{s} z_0'(\tilde{s}) ) w C_2(s) \,.
\end{eqnarray*}
For this to give an equation for $z_1$ as a function of the self-similar variable $\tilde{s}$ we must have that $C_1(s) = c_3 s^3$ and $C_2(s) = c_2 s^2$.  To pin down $c_1$ and $c_2$ we expand the first 3 terms of the Taylor series for $z_1$, as given just above, and find that:
\begin{equation}
z_1(s) = -c_2 s^2 + \left(-72 c_2 + 6 c_3 + 810 \right) s^4 + \dots\,. 
\end{equation}
On the other hand from the asymptotic expansion (\ref{bnN-expansion}) we have 
\begin{equation}
z_1(s) = f_1( s, 1) = \frac{\partial^2}{\partial s_1^2} e_1\left(s_1, s\right) \bigg|_{s_1=0} \,,
\end{equation}
and thus the combinatorial meaning of the $j$th coefficient in the Taylor expansion of $z_1$ is the number of genus 1 maps, with 2 vertices of valence 1, and $j$ vertices of valence $3$. 
The set of maps with a fixed valence structure on their vertices is bijectively equivalent to the set of pairs of permutations $(\omega, \sigma)$, where $\sigma$ is a fixed permutation whose cycle structure matches the valence structure of the vertices and $\omega$ is a fixed point free product of disjoint transpositions, satisfying a further condition equivalent to connectedness of the corresponding maps.  It is straightforward to partition these pairs by the genus.  More details on this equivalence can be found in  \cite{BI} and \cite{zvonkin-how}.  This allows one to efficiently count the number  of maps corresponding to the first few Taylor coefficients of $z_1$. 
The numbers of genus 1 trivalent maps for $j=2$ and $4$ are found to be $0$ and $810 \cdot 4! = 19440$ respectively.  Therefore we conclude that $c_2 = c_3 = 0$.
$\Box$

\subsection{Toda Equations for $j = 3$}

In the trivalent case we find that the leading order equations become
\begin{eqnarray}
- \frac{d}{ds}h_0(s,w) &=& 3 \, \partial_w \left[f_0^2 + h_0^2 f_0\right]\\
- \frac{d}{ds}f_0(s,w) &=& 3 \,  \partial_w (h_0 f_0^2)  + 3 \, (\partial_w h_0) \left[f_0^2 + h_0^2 f_0\right]
\end{eqnarray}
 One could integrate these equations and, after determining the constants of integration, show that they are equivalent when $w=1$ to \eqref{ideal1} and \eqref{ideal2}.  As this has already been done for the leading order of the continuum difference string equations in section \ref{locsde} we omit the analogous computation here.

The higher order equations (for $g > 0$) are

\begin{eqnarray} 
\nonumber - \frac{d}{ds_3} 
\left(
\begin{array}{c}
  h_{2g}  \\
  f_g
\end{array}
\right) &=&
\label{Toda-triang} 3 \partial_w \left\{ \[
\begin{array}{cc}
2 h_0 f_0   &  (2 f_0 + h_0^2)   \\
 f_0 (2f_0 + h_0^2) &   2h_0 f_0
\end{array}
\] \left(
\begin{array}{c}
  h_{2g}  \\
  f_g
\end{array}
\right)
 \right\} + \\ \nonumber &&
 + 3 (2f_0 + h_0^2) \left(
\begin{array}{c}
  0  \\
  h_{0 w} f_g - f_{0 w} h_{2g}
\end{array}
\right) +
\left(
\begin{array}{c}
  \mbox{Forcing}_g^{(1)}  \\
  \mbox{Forcing}_g^{(2)}
\end{array}
\right)\,.
\end{eqnarray}
\medskip

\begin{rem} We now observe that, for the case of $g = 1$, by using Proposition \ref{prop.h2.f1} to re-express the terms in the second summand above in terms of $h_0, f_0$ and their $w$-derivatives, these terms may be absorbed into the forcing and those homogeneous terms that remain are now in pure conservation law form as was asserted in remark 
\ref{rem:3.5}. 
\end{rem}

\subsubsection{Odd Terms}

The odd terms of the expansion of $h(s, w)$ also generate a hierarchy of (scalar) quasi-linear pde, which are recurisevly decoupled from the even terms.  The odd terms do appear in the forcing terms for the non-homogeneous equations determining $h_{2g}$ and $f_g$ described in the previous subsection.

The $n^{-2g+1}$ term of the expansion of (\ref{h}) is
\begin{equation}
- \frac{dh_{2g-1}}{ds} = 3 \partial_w \left( 2 h_0 h_{2g-1} f_0 \right) + \mbox{Forcing}_{2g-1}\,.
\end{equation}

\section{Determining $e_g$} \label{sec:4}

Recalling the basic identity (\ref{Hirota})
\begin{equation}
b_n^2 = \frac{\tau^2_{n+1} \tau^2_{n-1}}{\tau^4_n} b_n^2(0) \,,
\end{equation}
we have, by taking logarithms, 
\begin{equation} \label{tauk-2nddiff} 
\log \tau^2_{n+1} - 2 \log \tau^2_n + \log \tau^2_{n-1} = \log(b^2_n) - \log(b^2_n)(0)\,,
\end{equation}
where the initial value $b^2_n(0) = n$ is given by the recursion relations  of the Hermite polynomials.  
As in \cite{EMP08}, we can use formula \eqref{tauk-2nddiff} to recursively determine $e_g$ in terms of solutions to the continuum equations.  We use the asymptotic expansion of $b^2_n$ which has the form (\ref{f1}):
\begin{align} \nonumber
\frac{1}{n} b_n^2 &= \sum_{g=0}^\infty  f_g(s,1) n^{-2g} \\
&= \sum_{g=0}^\infty z_g(s) n^{-2g}\,,
\end{align}
where we have used the self-similar scaling:
\begin{equation}
f_g( s, w) = w^{1-2g} z_g(s w^{1/2}) \,, \quad \mbox{giving} \quad 
f_g(s, 1) = z_g(s) \,.
\end{equation}
In this section, unless otherwise stated we will use $s$ to denote $s_3$.
It should also be noted that the left hand side of equation (\ref{tauk-2nddiff}) has the form of a centered second difference, $\Delta_1 \tau^2_{n,n} - \Delta_{-1}\tau^2_{n.n}$.

We introduce here the classes of \emph{iterated integrals of rational functions} (or iir for short).  These classes are defined inductively in terms of the variable $z = z_0$ regarded as an independent variable.  To begin with, the class contains rational functions of $z$.  One then adds integrals of these rational functions with respect to $dz$.  Next one considers the vector space of polynomials in products of these integrals over the field of rational functions of $z$ and augments the space by integrals, with respect to $dz$, of these functions.  One continues this iterative process up to any given finite stage.  (In the classical literature these classes are sometimes referred to as {\it abelian functions}.)

It follows recursively from \eqref{IntFact} and \eqref{IntFact-odd} that $u_{2g}, u_{2g-1}$ and $z_g$ are possibly in a larger class of functions given by \emph{iterated integrals of rational functions as well as square roots of rational functions} of $z_0$.

For general $g$ we have the theorem:
\begin{thm} \label{thm5.2}
The function $e_g(s_3)$ satisfies
\begin{align}\label{thm2.2-eq}
e_g(s_3) &= \frac{4}{\gamma_2+1} \left[ s_3^{-\gamma_1 -1} \int_0^{s_3} s^{\gamma_1} H_g(s) ds - 
s_3^{-\gamma_1-\gamma_2-2} \int_0^{s_3} s^{\gamma_1+\gamma_2+1} H_g(s) ds \right] +
\\ &\phantom{=}\nonumber + C_1 s_3^{-\gamma_1-1} + C_2 s_3^{-\gamma_1-\gamma_2-2}\,, 
\end{align}
where $H_g(s)$ is a collection of recursively defined drivers for $e_g$ involving terms depending on
 $z_j$ for $j\leq g$ and $e_j$ for $j<g$, 
with 
\begin{equation} \label{gamma1-gamma2}
(\gamma_1, \gamma_2) = (1-4g, 1) \quad \mbox{or} \quad (3-4g, -3)\,,
\end{equation}
and where $C_1$ and $C_2$ are constants of integration determined either by the analyticity of $e_g(s_3)$ or the initial Taylor coefficients of $e_g(s_3)$ determined by some other method (for instance direct counting of maps with few vertices).  Either choice for the pair $(\gamma_1, \gamma_2)$ in \eqref{gamma1-gamma2} produce the same expression.  Moreover, assuming that $z_j$ for $j \leq g$ are class iir, $e_g(s_3)$
is also in the class of iterated integrals of rational functions.
\end{thm}

Here $s$ is a variable of integration, although it plays, in the integrand, the role of $s_3$.

\begin{proof}
We start from the expression (\ref{tauk-2nddiff}) and inductively assume that all necessary $z_j$ have been determined (we will need them for $j\leq g$, and will also need that, by induction, $e_j$ has been determined for $j<g$).
The left hand side of (\ref{tauk-2nddiff}) is a second order centered difference and so has an expansion for large $n$ involving only even derivatives of the spatial variable $w$.  We have 
\begin{align}
\sum_{g \geq 0} \frac{1}{n^{2g}} &\left[ \frac{\partial^2}{\partial w^2} w^{2-2g} e_g + 
\frac{1}{12} \frac{\partial^4}{\partial w^4} w^{4-2g} e_{g-1} + \dots + 
\frac{2}{(2g)!} \frac{\partial^{2g}}{\partial w^{2g}} w^2 e_0 \right]_{w=1} = \log( z_0 ) +
\nonumber \\ & \qquad 
+ \sum_{j=1}^\infty \frac{1}{n^{2j}} \left[ \frac{z_j}{z_0} - \frac{z_{j-1} z_1}{z_0^2} + \dots 
+ \frac{(-1)^{j+1}}{j} \frac{z_1^j}{z_0^j} \right] 
\end{align}
where on the left hand side $e_h = e_h(w^{1/2} s_3)$.

In the coefficient of $n^{-2g}$ on the left hand side one expands the term containing $e_g$ as a second order linear differential operator applied to $e_g$.  The remaining contributions in the equation are terms that have been recursively determined. More precisely, expanding the left hand side of
\begin{equation}
\frac{\partial^2}{\partial w^2} w^{2-2g} e_g( w^{1/2} s_3) \bigg|_{w=1} = H_g(s_3) 
\end{equation}
we find
\begin{equation}
(2-2g) (1-2g) e_g + \frac{1}{4} (7 - 8 g) s_3 e_g' + \frac{1}{4} s_3^2 e_g'' = H_g(s_3) \,.
\end{equation}
We then multiply by $s_3^{\gamma_1}$ where 
\begin{equation}
\gamma_1 = 1 - 4 g\,, 3 - 4g  
\end{equation}
and integrate once to find 
\begin{equation}
\frac{(2-2g) (1-2g) }{\gamma_1 + 1} s_3^{\gamma_1+2} e_g + \frac{1}{4} s_3^{\gamma_1+2} e_g' = 
\int_0^{s_3} s^{\gamma_1} H_g(s) ds + C_1' \,.
\end{equation}
Next multiply both sides by $s_3^{\gamma_2}$ with 
\begin{equation}
\gamma_2 = 1\,, -3\,,
\end{equation}
respectively for each choice of $\gamma_1$,
and integrate once to find 
\begin{equation}
\frac{1}{4} s_3^{\gamma_1+\gamma_2 + 2} e_g = \int_0^{s_3} {s'}^{\gamma_2} \int_0^{s'} s^{\gamma_1} H_g(s) ds ds' + C_1 s_3^{\gamma_2+1} + C_2 \,.
\end{equation}
We conclude by switching the order of integration in the double integral and computing the integral with respect to $s'$.  

Finally we note that the $H_g(s_3)$ are functions of $z_j$ for $j\leq g$, and $e_j$ for $j<g$, and so provided that the $z_j$ are in the class of iir functions, we have recursively that $H_g$ is an iir function.  Therefore a consequence of formula \eqref{thm2.2-eq} and this assumption is that $e_g$ will be an iir function.

\end{proof}

\begin{cor}
If the driver terms have Taylor expansion
\begin{equation} H_g(s_3) = \sum_{k=1}^\infty \eta_g(2k) s_3^{2k} \end{equation}
then the Taylor coefficients of the $e_g(s_3)$ take the form
\begin{equation} \kappa^{(3)}_g(2k) = \eta_g(2k) \frac{(2k)!}{(1-2g+k) (2-2g+k)} \end{equation}
for $k \neq 2g -1, 2g-2$\,.
\end{cor}

The proof follows from integration using the power rule. 
The exceptions, or possible resonances, at $k = 2g-1$ and $2g-2$ occur at precisely the powers of $s$ associated with the constants of integration $C_1$ and $C_2$. To demonstrate the usefulness of Theorem \ref{thm5.2}, we will now compute explicity expressions for $e_0$ and $e_1$ as functions of the fundamental auxiliary variable $z_0$.  These are analogous to the expressions determined for the $e_g$ in \cite{EMP08}.

\subsection{Example: $g=0$}

In the case of $g=0$, we have that $H_0(s) = \log(z_0)$, and Theorem \ref{thm5.2} gives 
\begin{equation}
e_0(s_3) = 2 \left[ s_3^{-2} \int_0^{s_3} s \log(z_0) ds - s_3^{-4} \int_0^{s_3} s^3 \log(z_0) ds \right] 
+ C_1 s_3^{-2} + C_2 s_3^{-4} \,.
\end{equation}
We see immediately that the constants must be zero to preserve analyticity near $s_3 = 0$.

Following \cite{EMP08} we write both integrals in terms of $z_0$, determined as a function of $s$ by equation (\ref{g_eqn}):
\begin{equation}\label{resultant}
1 = z_0^2 - 72 s^2 z_0^3\,.
\end{equation}
One can then solve this equation for $s$ as a function of $z_0$:
\begin{equation}\label{ssub}
s = \sqrt{ \frac{z_0^2 -1}{72 z_0^3} } \,,
\quad
s ds = -\frac{1}{144} \frac{(z_0^2-3)}{z_0^4} dz_0\,.
\end{equation}
Thus we find:
\begin{align} \nonumber
e_0(s_3) &= - \frac{z_0^3}{(z_0^2-1)} \int_1^{z_0} \frac{(z^2-3)}{z^4} \log(z) dz + 
\frac{z_0^6}{(z_0^2-1)^2} \int_1^{z_0} \frac{(z^2-1) (z^2-3)}{z^7} \log(z) dz
\\ \nonumber  &= - \frac{z_0^3}{(z_0^2-1)} \left[ -\frac{(z_0^2-1)}{z_0^3} \log(z_0) + \frac{1}{3} \frac{(z_0-1)^2 (2 z_0+1)}{z_0^3} \right] +
\\ \nonumber &\phantom{=}\hspace{1cm}
+ \frac{z_0^6}{(z_0^2-1)^2} \left[ -\frac{1}{2} \frac{(z_0^2-1)^2}{z_0^6} \log(z_0)
+ \frac{1}{12} \frac{(z_0^2-1)^3}{z_0^6} \right] 
\\ &= \frac{1}{2} \log(z_0) + \frac{1}{12} \frac{(z_0-1)(z_0^2 - 6 z_0 - 3)}{(z_0+1)} \,. \label{e0}
\end{align}

\subsection{Example: $g=1$}

In the case of $g=1$, we have that 
\begin{equation}
H_1(s) = \frac{z_1}{z_0} - \frac{1}{12} \frac{\partial^4}{\partial w^4} w^2 e_0(s w^{1/2}) \bigg|_{w=1} \,,
\end{equation}
where $z_1(s)$ is given by Proposition \ref{prop.h2.f1} as follows:  from that proposition, after a bit of manipulation using the first component of equation (\ref{leading-order-system}), we have 
\begin{equation}
f_1( s, w) = w^{-1} z_1(s w^{1/2}) = - \frac{3}{2} s f_0 h_{0, w} f_{0, ww} - \frac{3}{4} s f_0 h_{0, w}^3 \,,
\end{equation}
into which we subsitute the self-similar scalings $f_0 = w z_0(s w^{1/2})$ and $h_0 = w^{1/2} u_0(s w^{1/2})$ ( see Theorem \ref{thm:hf} ), expand the $w$-derivatives and then set $w=1$ to find
\begin{equation}
z_1(s) = - \frac{3}{2} s z_0 \left( \frac{1}{2} u_0 + \frac{1}{2} s u_0' \right) \left( \frac{3}{4} s z_0' + \frac{1}{4} s^2 z_0'' \right) - \frac{3}{4} s z_0 \left( \frac{1}{2} u_0 + \frac{1}{2} s u_0' \right)^3 \,;
\end{equation}
finally we use the algebraic relations (\ref{alg.1}) and (\ref{alg.2}), and the expression (\ref{ssub}) to eliminate all but $z_0$ from the formula for $z_1$, and we have
\begin{equation} \label{z1}
z_1(s) = \frac{1}{4} \frac{(z_0^2 - 1)^2 (z_0^2 + 9) z_0 }{(z_0^2-3)^4} \,.
\end{equation}
Theorem \ref{thm5.2} gives 
\begin{equation}
e_1(s_3) = 2 \left[ s_3^2 \int_0^{s_3} s^{-3} H_1(z_0) ds - \int_0^{s_3} s^{-1} H_1(z_0) ds \right] 
+ C_1 s_3^2 + C_2\,. 
\end{equation}
We make the change of variables, as before, to integrals with respect to $z_0$ using (\ref{ssub}).  We then have: 
\begin{eqnarray} \label{pre-e1}
e_1(s_3) &=& 
- \frac{(z_0^2-1)}{z_0^3} \int_1^{z_0} \frac{z^2 (z^2-3)}{(z^2-1)^2} H_1(z) dz 
+ \int_1^{z_0} \frac{(z^2-3)}{z (z^2-1)} H_1(z) dz +
\\ && \hspace{4cm}
+ C_1 s_3^2 + C_2. \nonumber
\end{eqnarray}
A direct calculation of $H_1(z_0)$ using (\ref{resultant}), (\ref{e0}) and (\ref{z1}) shows that the integrands of both integrals in (\ref{pre-e1}) are regular at $z_0 = 1$ (which corresponds to $s_3 = 0$). 
Hence, the vanishing of $e_1(0)$, which follows from (\ref{I.002}), implies that $C_2 = 0$. To determine $C_1$ we will need to appeal to the Taylor expansion of $e_1$ and its combinatorial interpretation.

Using the explicit form of $H_1$ as a  function of $z_0(s_3)$ one finds from (\ref{pre-e1}) that
\begin{eqnarray}
e_1(s_3) &=& - \frac{1}{24} \log\left( \frac{3}{2} - \frac{z_0^2}{2} \right) + \left( \frac{C_1}{72} - \frac{1}{48} \right) \frac{(z_0^2-1)}{z_0^3} 
\\ &=& - \frac{1}{24} \log\left( \frac{3}{2} - \frac{z_0^2}{2} \right) + \left( C_1 - \frac{3}{2} \right) s_3^2
\\ \label{e1} &=& - \frac{1}{24} \log\left( \frac{3}{2} - \frac{z_0^2}{2} \right)\,, 
\end{eqnarray}
where we have choosen $C_1 = 3/2$ so that the coefficient of $s_3^2$ in $e_1(s_3)$ will be $3/2!$, as can be checked by directly calculating the second derivative of (\ref{e1}) with respect to $s_3$ using the differential relation in (\ref{ssub}) and then evaluating at $s_3 = 0$ (equivalently $z_0 =1$). This value of the second order coefficient is required by the fact that there are three genus one maps with two vertices of valence 3.

The procedure can be continued for as far as one wishes; the only real constraint is the ability to find explicit expressions for the $z_g$ needed.  We state without the computations the formula derived for $e_2(s_3)$:
\begin{equation}
e_2(s_3) = \frac{1}{960} \frac{(z_0^2-1)^3 (4 z_0^4 - 93 z_0^2 - 261)}{(z_0^2-3)^5} \,.
\end{equation}
The formulas we have derived for $e_0, e_1,$ and $e_2$ as functions of $z_0$ have much in common with those found in the case of even times in \cite{EMP08}.  This suggests an extension of the global result, proven for the case of even times in \cite{Er09}); namely, we expect that 
for $g > 1$, $e_g(s_3)$ is a rational function of $z_0^2$ with singularities occurring only  at $z_0^2 = 3$, where $z_0$ is related to $s_3$ by 
\begin{equation*}
1 = z_0^2 - 72 s_3^2 z_0^3 \,.
\end{equation*}

\subsection{Taylor Coefficients of $e_0(s_3)$ and $e_1(s_3)$}

From formulas for $e_g$ in terms of $z_0$ it is a straightforward procedure to derive expressions for the Taylor coefficients using contour integrals.  The trick as in \cite{EMP08} is to again change variables to integrals with respect to $z_0$.  

One represents the Taylor coefficients of $e_0(s_3)$  as contour integrals and substitutes the expression (\ref{e0}) in terms of $z_0$.  The $2j$th Taylor coefficient is given by:
\begin{align}
\frac{K_{2j}^{(0)}}{(2j)!} &= \frac{1}{2\pi i} \oint_{s_3\sim 0} \frac{e_0}{s_3^{2j+1}} ds_3 \\ \label{3.36}
&=\frac{1}{j} \frac{1}{2\pi i} \oint_{z \sim 1} \frac{de_0}{dz} s_3^{-2j} dz \\
&=\frac{1}{j} \frac{1}{2\pi i} \oint_{z \sim 1} \frac{(z^2-3)(z^2-2z-1)}{6 z(z+1)^2} \left[ 72 \frac{z^{3}}{(z^2-1)} \right]^{j} dz \\
&=\frac{3^{2j-1} 2^{3j-2}}{j} \frac{1}{2\pi i} \oint_{z \sim 1} z^{3j-1} (z-1)^2 (z^2-3) (z^2-2z-1) \frac{dz}{(z^2-1)^{j+2}} \,,
\end{align}
where we have used an integration by parts to find line (\ref{3.36}), and 
the notation $\displaystyle \oint_{z\sim 1} $ indicates that the integral is over a small circle in the complex plane, oriented counter clockwise, and containing $z = 1$.
We then note that 
\begin{equation} (z-1)^2 (z^2-3) (z^2-2z-1) = z^6 - 4 z^5 + z^4 + 12 z^3 - 13 z + 3 \,.\end{equation}
We insert this into the integral together with the change of variables $\zeta = z^2$:
\begin{align} \nonumber
\frac{K_{2j}^{(0)}}{(2j)!} &= \frac{3^{2j-1} 2^{3j-2} }{j} \frac{1}{2\pi i} \oint_{\zeta \sim 1} 
\zeta^{(3j-2)/2} \left( \zeta^3 - 4 \zeta^{5/2} + \zeta^2 + 12 \zeta^{3/2} - 13 \zeta + 3\right) \frac{d\zeta}{(\zeta-1)^{j+2}} \\ \nonumber
&= \frac{3^{2j-1} 2^{3j-2} }{j} \left( \binom{\frac{3j}{2} +2}{j+1} - 4 \binom{\frac{3j}{2} + \frac{3}{2} }{j+1} + \binom{\frac{3j}{2} +1}{j+1} + 12 \binom{\frac{3j}{2} +\frac{1}{2}}{j+1} \right.
\\&\phantom{=} \left. \qquad\qquad\qquad\qquad \nonumber
- 13 \binom{\frac{3j}{2}}{j+1} + 3 \binom{\frac{3j}{2}-1}{j+1} \right) \\
&= \frac{3^{2j} 2^{3j}}{j} \frac{\Gamma\left(\frac{3j}{2}\right)}{\Gamma\left(\frac{j}{2}\right) \Gamma\left( 3 +j\right)} \,.
\end{align}

Likewise we can express the Taylor coefficients of the genus one expansion as contour integrals:
\begin{equation}
\frac{K_{2j}^{(1)}}{(2j)!} = \frac{3^{2j-1} 2^{3j-2}}{j 2\pi i} 
\oint_{\zeta \sim 1} -\frac{\zeta^{(3j+1)/2}}{(\zeta-3)} \frac{d\zeta}{(\zeta-1)^j}\,.
\end{equation}

\medskip

{\bf Acknowledgement.}  We thank the referees and the editor for their very careful reading of the manuscript.

\end{document}